\documentclass[a4paper,11pt,english,notitlepage,sumlimits]{article}

\usepackage[english]{babel}
\usepackage[utf8]{inputenc} 
\usepackage[a4paper,top=2cm,bottom=2.5cm,left=2.5cm,right=2.5cm,marginparwidth=1.75cm]{geometry} 
\usepackage{dsfont}
\usepackage{mathtools}
\usepackage{subfig}
\usepackage[arrowdel]{physics} 
\usepackage{multirow}  
\usepackage{xcolor}
\usepackage{enumitem}
\usepackage{amssymb}
\usepackage{fancyhdr}
\usepackage{lmodern}
\usepackage[T1]{fontenc}
\usepackage{caption}
\usepackage{float}
\usepackage{cancel} 
\usepackage{graphics} 
\usepackage{amsthm}
\usepackage{revsymb}
\usepackage{authblk}
\usepackage{amsmath}
\usepackage[style=ieee,giveninits=false,sorting=none,maxnames=10,minnames=3,dashed=false]{biblatex}
\usepackage{csquotes}

\usepackage[colorlinks=true, allcolors=black]{hyperref} 
\usepackage{orcidlink}
\numberwithin{equation}{section}

\newtheorem{theorem}{Theorem}[section]
\newtheorem{corollary}[theorem]{Corollary}
\newtheorem{lemma}[theorem]{Lemma}
\newtheorem{definition}[theorem]{Definition}

\newtheorem{remark}[theorem]{Remark}

\newlength{\blank}
\newenvironment{proofof}[1][{\hspace{-\blank}}]{{\medskip\noindent\textit{Proof~{#1}}.\ }}{\qed}
\newcommand{\CC}{\mathbb{C}}
\newcommand{\EE}{\mathbb{E}}

\newcommand{\cD}{\mathcal{D}}
\newcommand{\cE}{\mathcal{E}}
\newcommand{\cP}{\mathcal{P}}
\newcommand{\cR}{\mathcal{R}}
\newcommand{\cS}{\mathcal{S}}
\newcommand{\cT}{\mathcal{T}}
\newcommand{\cU}{\mathcal{U}}
\newcommand{\cN}{\mathcal{N}}
\newcommand{\1}{\openone} 
\newcommand{\id}{\operatorname{id}} 
\newcommand{\ox}{\otimes}
\newcommand{\proj}[1]{\ket{#1}\!\bra{#1}}
\begin{document}

\title{Decoupling by local random unitaries \protect\\ without simultaneous smoothing, \protect\\ and applications to multi-user quantum information tasks}

\author[1,2,4,$\ast$, \orcidlink{0000-0002-0126-4521}]{Pau Colomer}
\author[1,3,4,$\dagger$ \orcidlink{0000-0001-6344-4870}]{Andreas Winter}

\affil[1]{{\small\it Grup d'Informaci\'o Qu\`antica, Departament de F\'isica, \protect\\ Universitat Aut\`onoma de Barcelona, 08193 Bellaterra (Barcelona), Spain\vspace{0mm}}}
\affil[2]{{\small\it Lehrstuhl f\"ur Theoretische Informationstechnik, School of Computation, Information and Technology, Technische Universit\"at M\"unchen, Theresienstra{\ss}e 90, 80333 M\"unchen, Germany\vspace{1mm}}}
\affil[3]{{\small\it ICREA--Instituci\'o Catalana de Recerca i Estudis Avan\c{c}ats, \protect\\ Pg.~Lluis Companys, 08010 Barcelona, Spain\vspace{1mm}}}
\affil[4]{{\small\it Institute for Advanced Study, Technische Universit\"at M\"unchen, \protect\\ Lichtenbergstra{\ss}e 2a, 85748 Garching, Germany\vspace{0mm}}}
\affil[$\ast$]{{\small\textit{Email:} pau.colomer@tum.de} {\small\textit{ORCID:} 0000-0002-0126-4521}}
\affil[$\dagger$]{{\small\textit{Email:} andreas.winter@uab.cat} {\small\textit{ORCID:} 0000-0001-6344-4870}}

\date{}

\maketitle

\begin{abstract}
We show that a simple telescoping sum trick, together with the triangle inequality and a tensorisation property of expected-contractive coefficients of random channels, allow us to achieve general simultaneous decoupling for multiple users via local actions. Employing both old [Dupuis \emph{et al.}, Commun. Math. Phys. 328:251-284 (2014)] and new methods [Dupuis, IEEE Trans. Inf. Theory 69:7784-7792 (2023)], we obtain bounds on the expected deviation from ideal decoupling either in the one-shot setting in terms of smooth min-entropies, or the finite block length setting in terms of R\'enyi entropies. These bounds are essentially optimal without the need to address the simultaneous smoothing conjecture, which remains unresolved.
This leads to one-shot, finite block length, and asymptotic achievability results for several tasks in quantum Shannon theory, including 
local randomness extraction of multiple parties, 
multi-party assisted entanglement concentration, 
multi-party quantum state merging, 
and quantum coding for the quantum multiple access channel.
Because of the one-shot nature of our protocols, we obtain achievability results without the need for time-sharing, which at the same time leads to easy proofs of the asymptotic coding theorems. 
We show that our one-shot decoupling bounds furthermore yield achievable rates (so far only conjectured) for 
all four tasks in compound settings, which are additionally optimal for entanglement of assistance and state merging.
\end{abstract}


\section{Introduction}
\label{sec:intro}



Multi-user information theory is intrinsically difficult, with several of the classic transmission problems remaining unsolved despite decades of research, including the bidirectional channel \cite{Shannon-bidirectional}, the broadcast channel \cite{broadcast}, and the interference channel \cite{interference} (except in particular cases), cf.~\cite{ElGamal-Kim}. Even models such as the multiple-access channel (MAC) that have been solved early on \cite{Ahlswede:MAC,Liao:MAC} have recently exhibited unexpected additional complexity: indeed, while the capacity region of a general MAC has a finitary single-letter expression, its computation (or even approximation) in terms of the channel parameters turns out to be NP-hard, and with entanglement-assistance it is even uncomputable \cite{MAC:NP-hard}. 

The foundational tool of \emph{joint typicality} (a multipartite probability distribution being typical in several of its marginals simultaneously) is frequently employed in classical multi-user settings to define and analyze codes and decoders, and serves as a single conceptual integrator of many constructions (even if it does not always yield the best possible performance) \cite{CoverThomas,ElGamal-Kim}. 
The analogous problems in quantum information theory have added difficulty at an even more fundamental level because this basic tool is simply not available in the required generality for multipartite quantum states, although it has been conjectured both in a form suited to i.i.d systems \cite[Conj.~3.2.7]{Dutil:PhD} and in a general form for min-entropies \cite{DrescherFawzi:sim-min}, also known as the \emph{simultaneous smoothing conjecture}.
In the absence of a general solution to this conjecture (either in its one-shot or the asymptotic version), researchers have developed workarounds of varying complexity and applicability. For small numbers of parties (two or three) and specific problems, it can be avoided altogether \cite{DrescherFawzi:sim-min}; and for classical information transmission tasks over quantum channels with multiple senders and receivers, 
a ``simultaneous hypothesis testing'' technique combining a modification of the state with hypothesis testing \cite{Sen:sim-hypo-test,Sen:multi-user-one-shot} overcomes this technical barrier. 

However, there are at least two types of tasks that require a different primitive and therefore remain hindered by the simultaneous smoothing conjecture: cryptographic privacy amplification and randomness concentration on the one hand, and simultaneous quantum information transmission between multiple parties on the other (including channel coding, as well as channel simulation). All these impaired tasks can be based on the \emph{decoupling} of one part of a correlated state from another, via the concatenation of a unitary (typically random) and a fixed irreversible quantum channel. This primitive is well-developed in the case of a single system to decouple and well-understood to be governed by min-entropies \cite{SchumacherWestmoreland:decoupling, Renner:PhD, CaiWinterYeung, HOW:merging-Nature, HOW:merging-CMP, 
Devetak:Q,Abeyesinghe_2009, Berta:merging,Devetak:QRST,Bennett:QRST,Berta:QRST}. 

Here, we develop a solution for simultaneous decoupling, extending the ``generalised decoupling'' approach of Dupuis \emph{et al.} \cite{Dupuis-et-al:decouple} to multiple systems undergoing local random unitaries followed by a CPTP map (see Fig.~\ref{fig:general}). We are able to do so without addressing the simultaneous smoothing conjecture by leveraging contractivity properties of random channels and multiplicativity of contraction under tensor products. 
\begin{figure}[ht]
    \centering
    \includegraphics[scale=0.46]{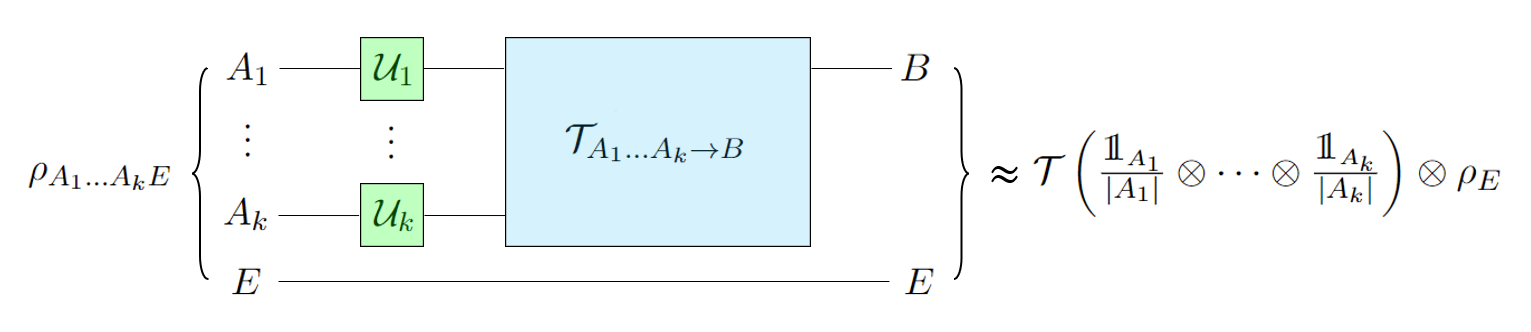}
    \caption{generalised multipartite decoupling via local random unitary transformations $\mathcal{U}_i$ acting locally on each system $A_i$, followed by a fixed CPTP map $\cT_{A_1\dots A_k\rightarrow B}$.}
    \label{fig:general}
\end{figure}
We illustrate the reach of our method by proving  multi-party generalised decoupling theorems in terms of both R\'enyi and smooth min-entropies. As applications, we show how we obtain one-shot and asymptotic i.i.d.~coding theorems for four quantum information tasks.

The rest of the paper is structured as follows: we start with some notation and preliminary material in Section \ref{sec:preliminaries}. Then we present the problem setting and main results in Section \ref{sec:setting}, followed by the core technical lemmas in Section \ref{sec:lemmata}. The main decoupling theorems are proved in Section \ref{sec:Proofs}. In Section \ref{sec:apps}, we apply the decoupling theorems to the problems of randomness extraction \cite{BertaFawziWehner:randomness,YHW:randomness}, entanglement of assistance \cite{SVW:EofA,HOW:merging-Nature,HOW:merging-CMP,Dutil:PhD}, quantum state merging (also known as quantum Slepian-Wolf problem), and quantum multiple access coding \cite{YDH:MAC}; these are developed in the fully general one-shot form, and then applied to the i.i.d.~asymptotics as well as to the so-called compound setting of an only partially known i.i.d.~source or channel. The resulting one-shot and compound rate formulas had long been conjectured but are here proved for the first time. 
We conclude in Section \ref{sec:conclusion}, including a comparison with previous approaches. 

\section{Preliminaries}
\label{sec:preliminaries}
We denote the Hilbert spaces associated with finite-dimensional quantum systems by capital letters, $A$, $B$, etc., and by $\abs{A}$ the dimension of $A$. The composition of two systems is facilitated by the tensor product of the Hilbert spaces, $AB=A\ox B$. Multipartite operators $\rho_{AB}$ acting on this  tensor product space have their corresponding reduced operator denoted as $\rho_A=\Tr_B \rho_{AB}$. The set of normalized quantum states (non-negative operators $\rho$ on $A$ with $\Tr\rho=1$) is denoted as $\cS(A)$.

We use the abbreviation CP to denote completely positive maps $\cT_{A\rightarrow B}$ (and CPTP if they are additionally trace-preserving maps), and $\tau_{AB}=\left(\1_A\ox\cT_{A'\rightarrow B}\right)\left(\ketbra{\Phi}{\Phi}_{AA'}\right)$ is their corresponding Choi operator, where $\ket{\Phi}_{AA'}=\frac{1}{\sqrt{\abs{A}}}\sum_{i=1}^{\abs{A}}\ket{i}_A\ox\ket{i}_{A'}$ is the maximally entangled state. 

The basic metric on quantum states is given by the trace norm distance, $\|\rho-\sigma\|_1$. 
Recall the definition of the trace norm of an operator $M$: $\|M\|_1=\Tr\sqrt{M^\dagger M}$. This quantity is lower bounded by 0 (when $\rho=\sigma$) and upper bounded by 2 due to the triangle inequality $\|\rho-\sigma\|_1\leq\|\rho\|_1+\|\sigma\|_1=2$. We shall mostly use the normalized trace distance defined as
\[
D(\rho,\sigma):=\frac12\|\rho-\sigma\|_1.
\]
We will also come across the Hilbert-Schmidt norm $\|M\|_2=\sqrt{\Tr M^\dagger M}$. Actually, it is useful to define the Schatten $p$-norms as a generalisation of the previous. Given a real number $p\geq1$ and a linear operator $M$, the Schatten $p$-norm is given by
\[
  \|M\|_p:=\left[\Tr\left(M^\dagger M\right)^\frac{p}{2}\right]^\frac{1}{p}.
\]

Likewise, we have to define the diamond norm of a linear map $\Theta:A\rightarrow B$, which is the trace norm of the output of a trivial extension of $\Theta$ maximized over all possible input operators $M\in A'A$ with $\|M\|_1\leq1$, that is $\displaystyle\|\Theta\|_\diamond=\max_{M \text{ s.t. } \|M\|_1\leq 1}\|(\id_{A'}\otimes\Theta)M\|_1$ \cite{Aharonov:diamond,Watrous:diamond}.

Our technical results are small upper bounds on the trace distance between states, proving that they are almost equal. These bounds are presented in terms of conditional entropy measures. Let us recall the following standard definitions. The von Neumann entropy of a state $\rho_A\in \cS(A)$ is defined as $S(A)_\rho = S(\rho_A) = -\Tr\rho\log\rho$, and the conditional von Neumann entropy of $A$ given $B$ for the bipartite state $\rho_{AB}$ is $S(A|B)_\rho=S(AB)_\rho-S(B)_\rho$. Also, for $\rho_{AB}\in \cS(AB)$ and $\sigma_B\in \cS(B)$, we define the sandwiched conditional R\'enyi entropy of order $\alpha\in[\frac{1}{2},1)\cup(1,\infty)$ given $\sigma_B$ \cite{Mueller-Lennert:sandwich,WWY:sandwich} as
\[
\widetilde{H}_\alpha(A|B)_{\rho|\sigma} 
  :=\begin{cases}
    \frac{1}{1-\alpha}\log\Tr\left[\left(\sigma_B^{\frac{1-\alpha}{2\alpha}}\rho_{AB}\sigma_B^{\frac{1-\alpha}{2\alpha}}\right)^\alpha\right] &\begin{split}&\text{if }\alpha<1 \text{ and } \Tr \rho\sigma \neq 0, \\
    &\text{or }\operatorname{supp}\rho \subseteq A\ox\operatorname{supp}(\sigma), \end{split}\\
    -\infty &\text{otherwise}.
  \end{cases}
\]

The maximisation of the sandwiched conditional R\'enyi entropy given $\sigma_B\in \cS(B)$ over all possible states $\sigma_B$ gives the conditional R\'enyi entropy of  $\rho_{AB}$, denoted $\widetilde{H}_\alpha(A|B)_\rho$. This quantity is monotone non-increasing in $\alpha$ \cite{Frank_2013}, and if we take the limit $\alpha\rightarrow1$ we recover the conditional von Neumann entropy $S(A|B)_\rho$. Furthermore, the limit of the R\'enyi entropy when $\alpha\rightarrow\infty$ makes sense and is called min-entropy:
\[
H_{\min}(A|B)_\rho = \widetilde{H}_\infty(A|B)_\rho := \max_{\sigma_B}\sup\{ \lambda\in\mathds{R}: \rho_{AB} \leq 2^{-\lambda}\cdot\1_A\ox\sigma_B\},
\]
where the maximum is taken over all states $\sigma_B\in\cS(B)$. Similarly, for $\alpha = \frac12$ we find the max-entropy:
\[
H_{\max}(A|B)_\rho = \widetilde{H}_{\frac12}(A|B)_\rho := \max_{\sigma_B} \log\left\|\sqrt{\rho_{AB}}\sqrt{\1\ox\sigma_B}\right\|_1^2.
\]

The max- and min-entropies are related by the fundamental duality relation $H_{\min}(A|B)_\psi=-H_{\max}(A|C)_\psi$ for any pure tripartite state $\psi_{ABC}$. Notice also that for $\alpha=2$ we find the collision entropy, which is the quantity that shows up in the original proofs of the variously general decoupling theorems \cite{Abeyesinghe_2009,Dupuis-et-al:decouple}:
\[
\widetilde{H}_2(A|B)_\rho = \sup_{\sigma_B} -\log\Tr\left[\left(\left(\1_A\ox\sigma_B^{-1/4}\right)\rho_{AB}\left(\1_A\ox\sigma_B^{-1/4}\right)\right)^2\right].
\]
This quantity, however, usually gives unreliable or loose bounds due to its rough responsiveness to small variations in the state $\rho_{AB}$ over which it is computed. This is why it is commonly substituted by the min-entropy, which is more reliable, and robust and a lower bound on the collision entropy due to the monotonicity of R\'enyi entropies under $\alpha$. In one-shot settings, it is also useful to $\epsilon$-smooth the min and max-entropies. I.e., computing them on the best state $\omega$ in an $\epsilon$-ball around $\rho$ with respect to the purified distance $P(\rho,\omega)=\sqrt{1-\|\sqrt{\rho}\sqrt{\omega}\|_1^2}$:
\[\begin{split}
  H_{\min}^\epsilon(A|B)_\rho:=\max_\omega H_{\min}(A|B)_\omega \text{ s.t. } P(\rho,\omega) \leq \epsilon,\\
  H_{\max}^\epsilon(A|C)_\rho:=\min_\omega H_{\max}(A|C)_\omega \text{ s.t. } P(\rho,\omega) \leq \epsilon.
\end{split}
\]

Smoothing allows us to discard atypical behaviour in the states. In multi-party settings, it makes sense to wish for \emph{simultaneous smoothing} of all the marginals of the given state: that is, we want to modify the global state so that its marginals appear smoothed. More formally, for any number $m$ of parties we would like to find functions $g_m(\epsilon)$ and $h_m(\epsilon)$ with ${\lim_{\epsilon\rightarrow0} g_m(\epsilon)=0}$ and $h_m(\epsilon)$ finite for any $m$ and $\epsilon>0$, such that for every state $\rho_{A_1\dots A_mB}$ on an $(m+1)$-party system $A_1\ldots A_mB$ there exists another state $\sigma$ with $P(\rho,\sigma) \leq g_m(\epsilon)$ that satisfies
\begin{equation}
  \forall\emptyset\neq I\subseteq[m] \quad H_{\min}(A_I|B)_{\sigma} \geq H_{\min}^{\epsilon}(A_I|B)_{\rho} - h_m(\epsilon).
\end{equation}
This has been stated as a conjecture \cite{DrescherFawzi:sim-min, FHSSW:inter-channel} (without the additive $h_m(\epsilon)$ term, which however seems very natural to us), but remains unproven in general, in particular for $m>2$. It has also been used to conjecture rate regions in several multi-party quantum information tasks. Here, we find local decoupling theorems without simultaneous smoothing and apply them to finally prove the anticipated achievable rate regions for several multi-party quantum information tasks. 

The purified distance between two arbitrary states $\rho$ and $\sigma$ is a function of the fidelity $F(\rho,\sigma)=\|\sqrt{\rho}\sqrt{\sigma}\|_1^2$, indeed $P(\rho,\sigma)=\sqrt{1-F(\rho,\sigma)}$. These quantities are related to the normalized trace distance 
through the Fuchs–van de Graaf inequalities \cite{fuchs+vandegraaf}:
\begin{equation}\label{eq:FvdG_relation}
  1-\sqrt{F(\rho,\sigma)} \leq D(\rho,\sigma) \leq P(\rho,\sigma)\leq\sqrt{D(\rho,\sigma)\left[2-D(\rho,\sigma)\right]}.
\end{equation}
The first two are the original inequalities. We took the liberty of adding the third one by noticing $ P(\rho,\sigma)^2=1-F(\rho,\sigma)\leq1-\left[1-D(\rho,\sigma)\right]^2=D(\rho,\sigma)\left[2-D(\rho,\sigma)\right]$.

\section{Setting and main results}
\label{sec:setting}
We consider random CP maps $\cR^x:A\rightarrow B$, where $x$ is distributed on a given set according to a certain well-defined probability law. If there are systems $A_1$, $A_2$, \ldots, $A_k$, we consider independently random maps $\cR^{x_i}:A_i\rightarrow B_i$ for $i\in\{1,\ldots,k\} =: [k]$. 
Two particular channels are of special interest to us. The the constant channel (or state preparation channel) $\cP^\sigma:A\rightarrow B$, acting as $\cP^\sigma(\rho) = \sigma_B\Tr_A\rho$, that outputs a state $\sigma$ (or more generally a positive semidefinite operator) on $B$ regardless of the input $\rho$, and the fully depolarizing channel $\cD:A\rightarrow A$, which can be seen as the particular instance of the constant channel that prepares the maximally mixed state $\cD(\rho)=\frac{\1_A}{|A|} \Tr_A\rho$.
We use superscripts to identify different objects, potentially acting on the same or other spaces, such as $\cR^{x_i}$ and $\cR^{x_j}$, and subscripts on states and channels to record on which systems they act. 

We shall only consider random CP maps $\cR^x$ with the property that the average map $\EE_x \cR^x$ is a constant map $\cP^\sigma$. Let us also introduce the difference $\Delta^x:=\cR^x-\cP^\sigma$. 
\begin{definition}\label{def:lambda}
We call a random Hermitian-preserving map $\Delta^x$ \emph{$\lambda$-expected-contractive} if for any system $E$ and any matrix $\rho_{AE}$,
\[
  \EE_x \|\Delta^x(\rho_{AE})\|_2 \leq \lambda \|\rho_{AE}\|_2.
\]
Dupuis \cite{Dupuis:Renyi-decouple} equivalently calls $\cR^x$ \emph{$\lambda$-randomizing}, although he considers this concept only for the preparation maximally mixed state $\sigma=\frac{1}{|B|}\1_B$.
\end{definition}

Let the systems $A_1$, $A_2$, \ldots, $A_k$ and $E$ share a state $\rho_{A_{[k]}E}$, and consider a fixed quantum channel (CPTP map) $\cT:A_{[k]}\rightarrow B$ with Choi state $\tau_{A_{[k]}B}$. On each system $A_i$ ($i\in[k]$) we define random unitaries $U_i$ distributed according to a unitary $2$-design, so that the average $\EE_{U_i} \cU_i = \cD_{A_i}$ is the completely depolarizing channel, where we denote the associated unitary channel $\cU_i(\alpha)=U_i\alpha U_i^\dagger$. Then we have random maps $\cR^{U_{[k]}} = \cT \circ (\cU_1\ox\cdots\ox\cU_k)$.

Decoupling is about the question: \emph{How far from $\tau_B = \cT\left({\1_{A_{[k]}}}/{|A_{[k]}|}\right)$ is the output of the channel $\cR^{U_{[k]}}$ typically?} 
To answer it, we aim to give an upper bound on
\[
  \EE_{U_{[k]}} \left\| \cR^{U_{[k]}}(\rho_{A_{[k]}E}) - \tau_B\ox\rho_E \right\|_1
  = \EE_{U_{[k]}} \left\| \cT \circ \left(\cU_1\ox\cdots\ox\cU_k - \cD_{A_1}\ox\cdots\ox\cD_{A_k}\right)\rho_{A_{[k]}E} \right\|_1.  
\]

The crucial insight for everything that follows is that we can rewrite the difference of maps inside the norm as 
\begin{equation}\begin{split}
  \label{eq:split-into-Theta}
  \cR^{U_{[k]}} - \cP^{\tau_B} 
    &= \cT\circ\left(\bigotimes_{i=1}^k \cU_i - \bigotimes_{i=1}^k \cD_{A_i} \right) \\
    &= \sum_{\emptyset\neq I\subseteq[k]} \cT\circ\left(\Theta_{A_I}\ox\cD_{A_{I^c}}\right),
\end{split}\end{equation}
where $\Theta_{A_i} := \cU_i-\cD_{A_i}$, hence $\Theta_{A_I}=\bigotimes_{i\in I}(\cU_i-\cD_{A_i})$, and $\cD_{A_{I^c}}=\bigotimes_{i\not\in I}\cD_i$. Therefore, we have $\cU_i = \Theta_{A_i} + \cD_{A_i}$ and we can use the distributive law to get the above expansion. Hence, 
\begin{equation}\begin{split}
  \label{eq:Theta}
  \left(\cR^{U_{[k]}} - \cP^{\tau_B}\right)\rho_{A_{[k]}E}
    &= \sum_{\emptyset\neq I\subseteq[k]} \cT\left((\Theta_{A_I}\ox\cD_{A_{I^c}})\rho_{A_{[k]}E}\right) \\
    &= \sum_{\emptyset\neq I\subseteq[k]} \left(\cT_I\circ\Theta_{A_I}\right)\rho_{A_IE}, 
\end{split}\end{equation}
with $\cT_I:A_I\rightarrow B$ acting as $\cT_I(\rho_{A_I}) = \cT\left(\rho_{A_I}\ox\frac{\1_{A_{I^c}}}{|A_{I^c}|}\right)$. The first step in our upper bound is the application of the triangle inequality, 
\begin{equation}
  \label{eq:Theta-triangle-bound}
  \EE_{U_{[k]}} \left\| \cR^{U_{[k]}}(\rho_{A_{[k]}E}) - \tau_B\ox\rho_E \right\|_1
  \leq \sum_{\emptyset\neq I\subseteq[k]} \EE_{U_I} \left\| \left(\cT_I\circ\Theta_{A_I}\right)\rho_{A_IE} \right\|_1.
\end{equation}
This allows us to simply deal with each term $\emptyset\neq I\subseteq[k]$ separately in the remainder of the argument.

The main technical results of the present work are formulated in the following theorems and their corollary. 

\begin{theorem}
\label{theorem:with_smoothing}
Assume $\cT_{A_{[k]} \rightarrow B}$ to be a CPTP map with Choi state, and consider the random channels $\cR^{U_{[k]}}$ as above. Then, for any state $\rho_{A_{[k]}E}$,
\begin{equation}
\begin{split}
\label{eq:with_smoothing}
\EE_{U_{[k]}} &\left\| \cR^{U_{[k]}}(\rho_{A_{[k]}E}) - \tau_B\ox\rho_E \right\|_1 \\
&\leq \sum_{\emptyset\neq I\subseteq[k]} \left\{ 2^{|I|+1}\epsilon_I + D_I \exp_2\left[-\frac{1}{2}\widetilde{H}_2^{\epsilon_I}(A_I|E)_{\rho|\zeta_E^I}-\frac{1}{2}\widetilde{H}_2(A_I|B)_{\tau|\sigma_B^I}\right] \right\},   
\end{split}    
\end{equation}
where $D_I = 2^{\abs{I}-1}\prod_{i\in I} \left(1-\frac{1}{\abs{A_i}^2}\right)^{-\frac{1}{2}}$, $\tau_B = \cT\left({\1_{A_{[k]}}}/{|A_{[k]}|}\right)$, the $\zeta_E^I$ are arbitrary states on $E$, $\sigma_B^I$ are arbitrary states on $B$, and $\exp_2$ denotes the exponential function to base $2$.
\end{theorem}

\begin{theorem}
\label{theorem:without_smoothing}
Assume $\cT:A_{[k]} \rightarrow B$ to be a CPTP map with $\cT(\1/|A_{[k]}|)=\1/|B|$, and consider the random channels $\cR^{U_{[k]}}$ as above. Then, for any state $\rho_{A_{[k]}E}$,
\begin{equation}\label{eq:withoutsmoothing}
\begin{split}
\EE_{U_{[k]}} &\left\| \cR^{U_{[k]}}(\rho_{A_{[k]}E}) - \tau_B\ox\rho_E \right\|_1 \\
&\leq \sum_{\emptyset\neq I\subseteq[k]} 
          D_I^{2-\frac{2}{\alpha_I}}
          2^{\frac{2}{\alpha_I}-1}
          \exp_2\left[ \left(1-\frac{1}{\alpha_I}\right) \left(-\widetilde{H}_{\alpha_I}(A_I|E)_{\rho|\zeta^I_E} - \widetilde{H}_2(A_I|B)_{\tau|\tau_B} \right)\right],
\end{split} 
\end{equation}
where $D_I = 2^{\abs{I}-1}\prod_{i\in I} \left(1-\frac{1}{\abs{A_i}^2}\right)^{-\frac{1}{2}}$ as before, 
$\alpha_I\in(1;2]$ are arbitrary real number numbers and $\zeta_E^I$ are arbitrary states on $E$. 
\end{theorem}

\begin{corollary}
\label{corollary:D_I-main}
Under the same conditions of Theorems \ref{theorem:with_smoothing} and \ref{theorem:without_smoothing}, and in the special case that the CP map is a tensor product, $\cT = \cT_1\ox \cdots \ox\cT_k$ with $\cT_i:A_i\rightarrow B_i$ (see Figure \ref{fig:tensorial}) and $B=B_1\dots B_k$, then Equations \eqref{eq:with_smoothing} and \eqref{eq:withoutsmoothing} hold with $D_I=\prod_{i\in I}\sqrt{1-\frac{1}{\abs{A_i}^2}}$.
\end{corollary}
\vspace{-5pt}
\begin{figure}[ht]
    \centering
    \includegraphics[scale=0.43]{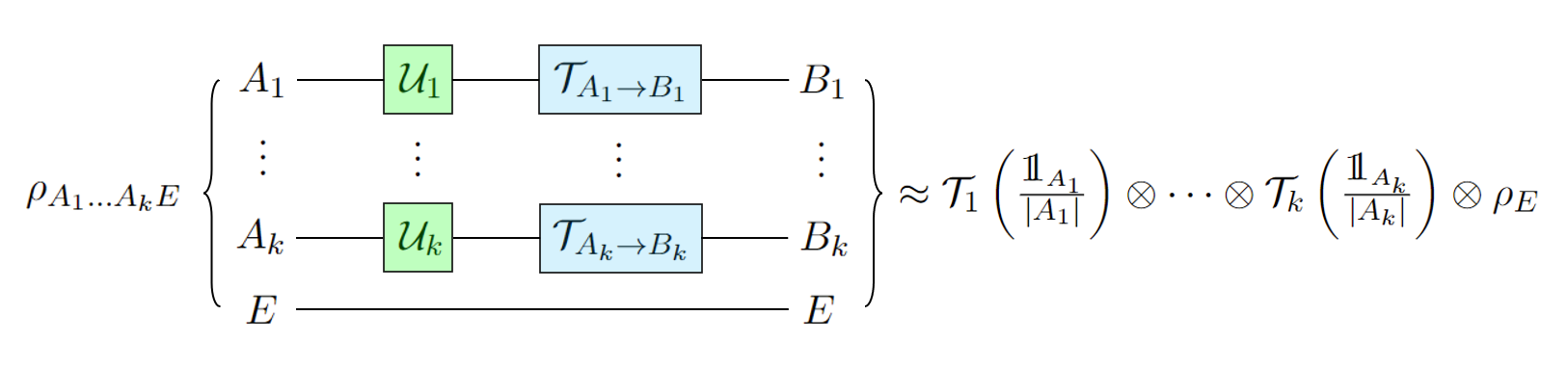}
    \caption{Multi-party decoupling via local random unitary transformations $\mathcal{U}_i$ followed by a fixed local CPTP map $\cT_{A_i\rightarrow B_i}$ on each of the systems $A_i$.}
    \label{fig:tensorial}
\end{figure}

\begin{remark}
\label{rem:2-vs-min}
\normalfont
In Theorem \ref{theorem:with_smoothing}, we will almost always use the lower bound $\widetilde{H}_2^{\epsilon_I}(A_I|E)_{\rho|\zeta_E^I} \geq H_{\min}^{\epsilon_I}(A_I|E)_{\rho|\zeta_E^I}$, and optimize $\zeta_E^I$ for the min-entropy, so that the first term in the exponential of Equation~\eqref{eq:with_smoothing} becomes $H_{\min}^{\epsilon_I}(A_I|E)_{\rho}$.
\end{remark}

\begin{remark}
\label{rem:single-system}
\normalfont
For $k=1$, both the above theorems, or more precisely their versions from Corollary \ref{corollary:D_I-main}, reproduce well-known predecessors:
Theorem \ref{theorem:with_smoothing} is essentially the general decoupling theorem from \cite{Dupuis-et-al:decouple}, albeit without the smoothing of the channel Choi matrix $\tau_{AB}$ (which in practice seems less critical than that of the state). 
Theorem \ref{theorem:without_smoothing} is a restatement of the main result of \cite{Dupuis:Renyi-decouple}; see also the precursor \cite{Mojahedian-et-al:correlation}. 
\end{remark}

\begin{remark}
\label{rem:Hao-Chung-and-friends}
\normalfont
Hao-Chung Cheng, Li Gao, and Mario Berta, in concurrent and independent work \cite{ChengGaoBerta:broadcast}, have discovered the same telescoping trick to obtain similar decoupling bounds, and in fact a multipartite version of the convex-split lemma. In their work, they apply the latter to the simulation of broadcast channels. 
\end{remark}

\section{Lemmata}
\label{sec:lemmata}
\subsection{Technical ingredients for the proofs}
In this subsection, we collect some well-known technical lemmas that will be used throughout the paper. Their proofs can be found in \cite{Dupuis-et-al:decouple}.

\begin{lemma}
\label{lemma:tr_norm_bound}
Let $M$ be a linear operator on $A$ and $\sigma$ a positive defined operator. Then
\begin{equation}
    \| M \|_1\leq \sqrt{\Tr[\sigma]\cdot\Tr[\sigma^{-1/4}M\sigma^{-1/2}M^\dagger \sigma^{-1/4}]}.
\end{equation}

If $M$ is Hermitian this can be simplified to
\begin{equation}
     \| M \|_1\leq \sqrt{\Tr[\sigma]\cdot\Tr[(\sigma^{-1/4}M\sigma^{-1/4})^2]}.
\end{equation}
\end{lemma}

\begin{lemma}
\label{lemma:swaps}
Let $M$ and $N$ be two linear operators on $A^{\otimes2}$, and let $F_A$ swap the two copies of the $A$ system: $F_A\left(\sum_{ij}m_{ij}\ket{i}\ket{j}\right)=\sum_{ij}m_{ij}\ket{j}\ket{i}$. Then,  $\Tr(M\otimes N)F_A = \Tr MN$.
\end{lemma}

\begin{lemma}
\label{lemma:Haar}
Let $M$ be a linear operator acting on the Hilbert space $A^{\otimes 2}$. Then, for a random unitary $U$ distributed according to a $2$-design (a set of unitaries that collectively approximates the statistical behavior of the entire unitary group up to second-degree polynomials \cite{DCEE:2-designs}),
\begin{equation}
    \EE_{U_A} \left( U^{\otimes 2}MU^{\otimes2\dagger} \right) = \alpha \1_{AA'}+\beta F_A,
\end{equation}
where $\alpha$ and $\beta$ are such that $\Tr M = \alpha\abs{A}^2+\beta\abs{A}$ and $\Tr MF = \alpha\abs{A}+\beta\abs{A}^2$.
\end{lemma}

We can easily generalise this lemma to a multipartite version:
\begin{corollary}
\label{corollary:schur_general}
Let $M$ be a linear operator acting on $A_1^{\otimes2}\ox\cdots\ox A_k^{\otimes2}$. Then, for $U_{[k]} = U_1 \ox\cdots\ox U_k$ the tensor product of independent unitaries distributed according to $2$-designs, 
\begin{equation}
  \EE_{U_{[k]}} \left( U_{[k]}^{\ox 2} M U_{[k]}^{\ox 2\dagger} \right)
  = \sum_{L\subseteq [k]} c_L\left(F_{A_L}\otimes\1_{A_{L^c}}^{\ox 2}\right),
\end{equation}
where $L^c=[k]\setminus L$ is the set complement of $L$, and the coefficients $c_L$ are determined by the relations
\begin{equation}
\Tr[M(F_{A_L}\otimes\1_{A_{L^c}}^{\ox 2})] = \sum_{T\subseteq [k]} c_T \abs{A_{T\cap L^c}} \abs{A_{T^c\cup L}}^2.
\end{equation}
\end{corollary}

\begin{lemma}
\label{lemma:partial_trace}
Let $\omega_{AB}$ be a non-negative operator acting on ${AB}$. Then, 
\begin{equation}
    \frac{1}{\abs{A}} \Tr \omega_{B}^2 
    \leq \Tr \omega_{AB}^2
    \leq \abs{A} \Tr \omega_{B}^2.
\end{equation}
\end{lemma}

\begin{lemma}
\label{lemma:tr_dist}
The normalized trace distance $D(\rho,\sigma)$ between two quantum states $\rho,\sigma\in \cS(A)$ is equal to the largest probability difference that the two states could give to the same measurement outcome $\Lambda$:
    \begin{equation}
        D(\rho,\sigma)=\max_{0\leq\Lambda\leq 1}\Tr\{\Lambda(\rho-\sigma)\}.
    \end{equation}
\end{lemma}

\begin{theorem}[Uhlmann's theorem for the purified distance]
\label{theorem:Uhlmann}
Let $\rho,\sigma\in \cS(A)$ and $\ket{\psi}\in A\ox A'$ be a purification of $\rho$, with $A'\cong A$. Then, there exists a purification $\ket{\phi}\in A\ox A'$ of $\sigma$ such that $P(\rho,\sigma)=P(\psi,\phi)$.
\end{theorem}

\subsection{Central lemmas}
The proofs of the main theorems (which we will present in the next section) rely on a series of new lemmas listed and proved in this section.
\begin{lemma}
\label{lemma:lambda_I}
Given any general CP map $\cT_I:A_{I}\rightarrow B$ with $I\subset [k]$, $\cT_I\circ\Theta_{A_I}$ is $\lambda_I$-expected-contractive with
\[
\lambda_I=\frac{2^{\abs{I}-1}}{\displaystyle\prod_{i\in I}\sqrt{1-\frac{1}{\abs{A_i}^2}}}\|\tau_{A_IB}\|_2.
\]
\end{lemma}
\begin{proof}
Recall that $\cT_I\circ\Theta_{A_I}$ is $\lambda_I$-expected-contractive if $\EE_{U_I}\|\left(\cT_I\circ\Theta_{A_I}\right)\rho_{A_IE}\|_2\leq\lambda_I\|\rho_{A_IE}\|_2$ for arbitrary $\rho_{AE}$, where $\EE_{U_I}$ is the expectation value over each ${U_i}$ with $i\in I$. We will start by assuming $\rho=\rho^\dagger$ to be Hermitian. Using Jensen's inequality, we find $(\EE_{U_I}\|\left(\cT_I\circ\Theta_{A_I}\right)\rho_{A_IE}\|_2)^2\leq\EE_{U_I}\|\left(\cT_I\circ\Theta_{A_I}\right)\rho_{A_IE}\|_2^2$. Now we can expand the square of the Hilbert-Schmidt norm as a trace without carrying the root throughout the demonstration:
\[
\EE_{U_I}\|\left(\cT_I\circ\Theta_{A_I}\right)\rho_{A_IE}\|_2^2=\EE_{U_I}\Tr\left[\left(\left(\cT_I\circ\Theta_{A_I}\right)\rho_{A_IE}\right)^2\right]=\Tr\left[\EE_{U_I}\left(\left(\cT_I\circ\Theta_{A_I}\right)\rho_{A_IE}\right)^2\right],
\]
where we have used the linearity of the trace in the second equality, and we are imposing for simplicity the additional restriction that the matrices $\rho_{A_IE}$ are Hermitian, and will prove in the end that the results can be generalised to any matrix. Defining a subset $J\subseteq I$ and its complement $J^c=I\setminus J$ we can write the expectation value as follows:
\begin{equation}\label{eq:AlternateSigns}
\EE_{U_{I}}\left((\cT_I\circ\Theta_{I})\rho_{A_IE}\right)^2=\sum_{ J\subseteq I}(-1)^\abs{J^c}\EE_{U_{J}}\left[\cT_I\circ(\cU_J\ox\cD_{A_{J^c}})\rho_{A_IE}\right]^2.
\end{equation}
We prove this claim by induction on the cardinality of $I$. For $\abs{I}=1$ (that is $A_I=A$) we have $
\EE_U\left[(\cT\circ\Theta_A)\rho_{AE}\right]^2=\EE_U[(\cR^U-\cT\circ\cD_{A})\rho_{AE}]^2$. Expanding the binomial and remembering our condition $\EE_U\cU(\rho_{AE})=\cD_A(\rho_{AE})$ we find:
\begin{equation*}
\begin{split}
\EE_U\left((\cT\circ\Theta_A)\rho_{AE}\right)^2&=\EE_U[\cR^U(\rho_{AE})^2]-2\EE_U[\cR^U(\rho_{AE})]\cdot[\cT\circ\cD_A(\rho_{AE})
]+[\cT\circ\cD_A(\rho_{AE})]^2\\
&=\EE_U[\cT\circ\cU(\rho_{AE})^2]-[\cT\circ\cD_A(\rho_{AE})]^2\\
&=\sum_{ J\subseteq I}(-1)^\abs{J^c}\EE_{U_{J}}\left[\cT_I\circ(\cU_J\ox\cD_{A_{J^c}})\rho_{A_IE}\right]^2,
\end{split}
\end{equation*}
because $\abs{J} \in \{0,1\}$. We continue the induction by assuming that Equation~\eqref{eq:AlternateSigns} is true for some $I$, and we want to pass to a bigger set $I'=I\stackrel{.}{\cup}\{i_0\}$ (with $\stackrel{.}{\cup}$ the disjoint union) to compute the expectation value $\EE_{U_{I'}}[(\cT_{I'}\circ\Theta_{A_{I'}})\rho_{A_{I'}E}]^2$ on $A_{I'}=A_I\otimes A_{i_0}$. Similarly let us define $J'=J\stackrel{.}{\cup}\{i_0\}$ for a subset $J\subseteq I$, that is $A_{J'}=A_J\otimes A_{i_0}$. Then, expanding $\Theta_{A_i}:=\cU_i-\cD_{A_i}$ and $\Theta_{A_I}=\bigotimes_{i\in I}(\cU_i-\cD_{A_i})$, and rearranging the terms (recall $\cD_{A_{I^c}}=\bigotimes_{i\not\in I}\cD_i$), we can re-write:
\begin{equation*}
\begin{split}
\EE_{U_{I'}}\left[(\cT_{I'}\circ\Theta_{A_{I'}})\rho_{A_{J'}E}\right]^2
 &=\EE_{U_{i_0}}\EE_{U_I}\left[\cT_{I'}\circ (\Theta_{A_I}\ox\Theta_{A_{i_0}})\rho_{A_{I'}E}\right]^2\\
 &=\sum_{ J\subseteq I}(-1)^\abs{J^c}\EE_{U_{J}}\left[\cT_{I'}\circ(\cU_J\ox\cD_{A_{J^c}})\rho_{A_{I'}E}\right]^2 \\
 &\phantom{=======:}
  \ox\EE_{U_{i_0}}\left[\cT_{I'}\circ(\cU_{{i_0}}-\cD_{A_{i_0}})\rho_{A_{I'}E}\right]^2.
\end{split}
\end{equation*}
By expanding the square (just as we did at the beginning of the induction) we can write $\EE_{U_{i_0}}\left[\cT_{I'}\circ(\cU_{{i_0}}-\cD_{A_{i_0}})\rho_{A_{I'}E}\right]^2=\EE_{U_{i_0}}\left[(\cT_{I'}\circ\cU_{{i_0}})\rho_{A_{I'}E}\right]^2-\left[(\cT_{I'}\circ\cD_{A_{i_0}})\rho_{A_{I'}E}\right]^2$. This allows us to write
\begin{equation*}
\begin{split}
\sum_{J\subseteq I} &(-1)^\abs{J^c}\EE_{U_{J}}\left[\cT_{I'}\circ(\cU_J\ox\cD_{A_{J^c}})\rho_{A_{I'}E}\right]^2\ox\EE_{U_{i_0}}\left[\cT_{I'}\circ(\cU_{{i_0}}-\cD_{A_{i_0}})\rho_{A_{I'}E}\right]^2\\
&=\sum_{ J\subseteq I}(-1)^\abs{J^c}\EE_{U_{J}}\EE_{U_{i_0}}\left[\cT_{I'}\circ(\cU_{J'}\ox\cD_{A_{J^c}})\rho_{A_{I'}E}\right]^2\\
&\phantom{====:}
 + (-1)^\abs{J'^c}\EE_{U_{J}}\left[\cT_{I'}\circ(\cU_J\ox\cD_{A_{J'^c}})\rho_{A_{I'}E}\right]^2\\
&= \sum_{ J\subseteq I'}(-1)^\abs{J^c}\EE_{U_{J}}\left[\cT_{I'}\circ(\cU_J\ox\cD_{A_{J^c}})\rho_{A_{I'}E}\right]^2.
\end{split}
\end{equation*}
This completes the proof by induction. Now we perform the trace:
\begin{equation*}
\begin{split}
\Tr\left[\EE_{U_I}\left(\left(\cT_I\circ\Theta_{A_I}\right)\rho_{A_IE}\right)^2\right] 
  &=\Tr\left[\sum_{ J\subseteq I}(-1)^\abs{J^c}\EE_{U_{J}}\left[\cT_I\circ(\cU_J\ox\cD_{A_{J^c}})\rho_{A_IE}\right]^2\right]\\
  &=\Tr\left[\sum_{ J\subseteq I}(-1)^\abs{J^c}\EE_{U_{J}}\cR^{U_I}\left(\rho_{A_JE}\ox\frac{\1_{A_{J^c}}}{\abs{A_{J^c}}}\right)^2\right]\\
  &= \Tr\left[\sum_{ J\subseteq I}(-1)^\abs{J^c}\EE_{U_{J}}\cR^{U_I}(\sigma_J)^2\right],
\end{split}
\end{equation*}
using the abbreviation $\sigma_J:= \rho_{A_JE}\ox\frac{\1_{A_{J^c}}}{\abs{A_{J^c}}}$ for ${J\subseteq I}$. 
Now we use the swap trick (as in \cite{Dupuis-et-al:decouple}) to simplify this expression:

\vspace{-5pt}
\begin{equation*}
\begin{split}
\Tr\left[\sum_{ J\subseteq I}(-1)^\abs{J^c}\EE_{U_{J}}\cR^{U_I}(\sigma_J)^2\right]
&=\Tr\left[\sum_{ J\subseteq I}(-1)^\abs{J^c}\EE_{U_{J}}\cR^{U_I}(\sigma_J)^{\ox2}F_{BE}\right]\\
&=\Tr\left[\sum_{ J\subseteq I}(-1)^\abs{J^c}\EE_{U_{J}}\left(\bigotimes_{i\in I}(U_i^\dagger)^{\ox2}(\cT^{\ox2})^\dagger F_{BE}(U_i)^{\ox2}\right)\sigma_J^{\ox2}\right]\\
&=\Tr\left[\sum_{ J\subseteq I}(-1)^\abs{J^c}\sigma_J^{\ox2}\left(\sum_{L\subseteq I}c_L(F_{LE}\ox\1_{L^c}^{\ox2})\right)\right]\\
&=\sum_{J,L\subseteq I}(-1)^\abs{J^c}\Tr[\rho_{A_{[J\cap L]} E}^2]c_L\prod_{i\in[J^c\cap L]}\frac{1}{\abs{A_i}},
\end{split}
\end{equation*}
where we have used Corollary \ref{corollary:schur_general} in the third equality. Notice that for a fixed $L$ we have $2^\abs{L}$ possible values for $\Tr[\rho_{A_{[J\cap L]} E}^2]\prod_{i\subseteq[J^c\cap L]}\frac{1}{\abs{A_i}}$. If we expand the sum, we find $2^{\abs{I}-\abs{L}}$ elements for each of the $2^\abs{L}$ possible values of the trace and product. Notice also that $2^{\abs{I}-\abs{L}-1}$ of this elements are positive and $2^{\abs{I}-\abs{L}-1}$ of them are negative. This implies that $\forall L\neq I$ the sum cancels. We just have to compute the case where $L$ is the whole $I$. We find:
\begin{equation}\label{eq:trace.approx}
\Tr\left[\EE_{U_I}\left(\left(\cT_I\circ\Theta_{A_I}\right)\rho_{A_IE}\right)^2\right]=c_I\sum_{J\subseteq I}\frac{(-1)^\abs{J^c}}{\abs{A_{J^c}}}\Tr[\rho_{A_JE}^2].
\end{equation}

To compute $c_I$ we follow the steps in \cite{chakraborty-et-al:multisender.decoupling}, let us define a new auxiliary subset $P\subseteq I$:
\[
\begin{bmatrix}
c_0 \\
\vdots \\
c_P \\
\vdots \\
c_I 
\end{bmatrix}=\frac{\abs{A_I}}{\displaystyle\prod_{i\in I}\left(\abs{A_i}^2-1\right)}\bigotimes_{i\in I} \begin{bmatrix}
\abs{A_i} & -1 \\
-1 & \abs{A_i} 
\end{bmatrix}\begin{bmatrix}
\Tr(\tau_B^2) \\
\vdots \\
\Tr(\tau_{A_PB}^2)  \\
\vdots \\
\Tr(\tau_{A_IB}^2) 
\end{bmatrix},
\]
thus
\begin{equation}\label{eq:c.approx}
\begin{split}
    c_I&=\frac{\abs{A_I}}{\displaystyle\prod_{i\in I}\left(\abs{A_i}^2-1\right)}\sum_{P\subseteq I}(-1)^{\abs{P^c}}\abs{A_P}\Tr(\tau_{A_PB}^2)\\
    &=\frac{1}{\displaystyle\prod_{i\in I}\left(1-\frac{1}{\abs{A_i}^2}\right)}\sum_{P\subseteq I} \frac{(-1)^{\abs{P^c}}}{\abs{A_{P^c}}}\Tr(\tau_{A_PB}^2).
\end{split}
\end{equation}

We can find an upper bound for Equation~\eqref{eq:trace.approx} by keeping only the positive terms of the sum, these are the terms such that $\abs{J^c}$ is even. And then transforming each partial trace $\Tr[\rho_{A_JE}^2]$ with $J\subseteq I$ to $\Tr[\rho_{A_IE}^2]$, using Lemma \ref{lemma:partial_trace} we find:
\[
  \sum_{J\subseteq I}(-1)^\abs{J^c}\frac{\Tr[\rho_{A_JE}^2]}{\abs{A_{J^c}}}
  \leq \!\!\sum_{\abs{J^c} \text{ even}}\!\! \frac{\Tr[\rho_{A_JE}^2]}{\abs{A_{J^c}}}
  \leq \Tr[\rho_{A_IE}^2] \left(\sum_{\abs{J^c} \text{ even}}\!\! 1\right)
  = 2^{\abs{I}-1}\Tr[\rho_{A_IE}^2],
\]
and similarly 
\[
  \sum_{J\subseteq I}(-1)^\abs{J^c}\frac{\Tr[\rho_{A_JE}^2]}{\abs{A_{J^c}}}
  \geq -\!\!\sum_{\abs{J^c}  \text{ odd}} \!\!\frac{\Tr[\rho_{A_JE}^2]}{\abs{A_{J^c}}}
  \geq -\Tr[\rho_{A_IE}^2] \left(\sum_{\abs{J^c} \text{ odd}}\!\! 1\right)
  = - 2^{\abs{I}-1}\Tr[\rho_{A_IE}^2],
\]
where we have used that any set $I$ has $2^{\abs{I}-1}$ subsets with an even number of elements, and the same number of subsets with an odd number of elements. 
With the same method we can bound $c_I$ from Equation~\eqref{eq:c.approx} and find 
\begin{align*}
c_I 
  =\frac{1}{\displaystyle\prod_{i\in I}\left(1-\frac{1}{\abs{A_i}^2}\right)}\sum_{L\subseteq I} \frac{(-1)^{\abs{P^c}}}{\abs{A_{P^c}}}\Tr(\tau_{A_PB}^2)
  &\leq \frac{2^{\abs{I}-1}}{\displaystyle\prod_{i\in I}\left(1-\frac{1}{\abs{A_i}^2}\right)}\Tr[\tau_{A_IB}^2], \\
c_I 
  &\geq - \frac{2^{\abs{I}-1}}{\displaystyle\prod_{i\in I}\left(1-\frac{1}{\abs{A_i}^2}\right)}\Tr[\tau_{A_IB}^2].
\end{align*}
Putting together these bounds we obtain
\begin{equation}\label{eq:lambda_hermitian}
  \left( \EE_{U_I} \left\|\left(\cT_I\circ\Theta_{A_I}\right)\rho_{A_IE}\right\|_2\right)^2
  \leq\EE_{U_I} \left\|\left(\cT_I\circ\Theta_{A_I}\right)\rho_{A_IE}\right\|_2^2
  \leq\frac{4^{\abs{I}-1}}{\displaystyle\prod_{i\in I}\left(1-\frac{1}{\abs{A_i}^2}\right)} \left\|\tau_{A_IB}\right\|_2^2 \left\|\rho_{A_IE}\right\|_2^2,
\end{equation}
which allows us to identify $\lambda_I^2\geq\left(4^{\abs{I}-1}/\prod_{i\in I}1-\frac{1}{\abs{A_i}^2}\right)\|\tau_{A_IB}\|_2^2$, under the restriction that $\rho_{A_IE}$ is Hermitian. 

It only remains to show how this result can be generalised from Hermitian to arbitrary matrices. To start, any matrix $\eta=\eta_{AE}$ can be decomposed into a Hermitian component $\eta_R$ and an anti-Hermitian component $i\eta_I$ (with $\eta_I$ Hermitian): $\eta=\eta_R+i\eta_I$. Now, applying the definition of the 2-norm we have
\[\begin{split}
  \left\|\Delta^U(\eta)\right\|_2^2
   = \Tr \left(\Delta^U(\eta)\right)^\dagger\!\Delta^U(\eta)
  &= \Tr \Delta^U(\eta_R-i\eta_I)\Delta^U(\eta_R+i\eta_I) \\
  &= \Tr\left[\Delta^U(\eta_R)^2+\Delta^U(\eta_I)^2\right],
\end{split}\]
where we have used the decomposition of $\eta$ (and $\eta^\dagger=\eta_R-i\eta_I$) in the second inequality, and the Hermitian preserving property of $\Delta^U=\cR^U-\cP^{\tau_B}$ and the distributive law in the last. With the linearity of the trace and the definition of the 2-norm it follows that
\(
\left\|\Delta^U(\eta)\right\|_2^2
  = \left\|\Delta^U(\eta_R)\right\|_2^2 
    + \left\|\Delta^U(\eta_I)\right\|_2^2.
\)
Now, notice that both $\eta_R$ and $\eta_I$ are Hermitian so we can employ Equation \eqref{eq:lambda_hermitian} to obtain the following bound:
\[
 \EE_U \left\|\Delta^U(\eta)\right\|_2^2 
 =\EE_U \left\|\Delta^U(\eta_R)\right\|_2^2 
 + \EE_U \left\|\Delta^U(\eta_I)\right\|_2^2
 \leq \lambda_I^2\|\eta_R\|_2^2 + \lambda_I^2\|\eta_I\|_2^2
 =\lambda_I^2\|\eta\|_2^2,
\]
where we have used $\|\eta\|_2^2 = \|\eta_R\|_2^2+\|\eta_I\|_2^2$ in the last step. This proves that the property $\EE_U\|\Delta^U(\rho)\|_2^2\leq\lambda^2\|\rho\|_2^2$ for Hermitian matrices extends to general matrices. So the inequality \eqref{eq:lambda_hermitian} holds for any matrix $\rho_{A_IE}$, and we can take the square root of this inequality to find
\(\EE_{U_I} \left\|\left(\cT_I\circ\Theta_{A_I}\right)\rho_{A_IE}\right\|_2\leq\lambda_I\|\rho_{A_IE}\|_2 \), concluding the proof.
\end{proof}

\begin{lemma}
\label{lemma:lambda.tensorial}
Consider a CP map $\cT:A\rightarrow B$
with Choi operator $\tau_{AB} = (\id\ox\cT)\Phi_{AA'}$. 
Define a random CP map $\cR^U := \cT(U\cdot U^\dagger)$, where $U$ is distributed according to a probability law $p$ on $SU(A)$ that is a $2$-design, for example, the Haar measure. 

Then, the family $\cR^U$ is $\lambda$-randomizing (equivalently, $\Delta^U=\cR^U-\cP^{\tau_B}$ is $\lambda$-expected-contractive) for any $\lambda \geq \sqrt{1-\frac{1}{\abs{A}^2}} \|\tau_{AB}\|_2$.
\end{lemma}

\begin{proof} 
Notice that this is actually nothing else than a simple particular case of Lemma \ref{lemma:lambda_I} with $\abs{I}=1$, this is $A_I=A$. We will just find a tighter bound on the constant $\lambda$. Let us again develop the argument first for Hermitian matrices. From Equations~\eqref{eq:trace.approx} and \eqref{eq:c.approx} we observe:
\[
\Tr\left[\EE_{U}\left(\left(\cT\circ\Theta\right)\rho_{AE}\right)^2\right]
 = c\left[\Tr \rho_{AE}^2 - \frac{\Tr
  \rho_{E}^2}{\abs{A}}\right], 
  \text{ with }
  c = \frac{\Tr \tau_{AB}^2 - \frac{\Tr \tau_{B}^2}{\abs{A}}}{1-\frac{1}{\abs{A}^2}}.
\]
We upper bound the parameter $c$ with the help of Lemma \ref{lemma:partial_trace}. 
Notice $-\abs{A}\Tr \tau_B^2 \leq -\Tr \tau_{AB}^2$, therefore $c\leq\Tr \tau_{AB}^2$. Similarly, $-\abs{A}\Tr \rho_B^2 \leq -\Tr \rho_{AB}^2$. We thus find
\[
\Tr\left[\EE_{U}\left(\left(\cT\circ\Theta\right)\rho_{AE}\right)^2\right]
 \leq \left(1-\frac{1}{\abs{A}^2}\right) (\Tr\tau_{AB}^2) (\Tr\rho_{AE}^2).
\]
Now, using Jensen's inequality we have
\[
(\EE_U\|(\cT\circ\Theta)\rho_{AE}\|_2)^2\leq\EE_U\|(\cT\circ\Theta)\rho_{AE}\|^2_2\leq\left(1-\frac{1}{\abs{A}^2}\right)\|\tau_{AB}\|_2^2\|\rho_{AE}\|_2^2,
\]
Repeating the argument that concluded the proof of Lemma \ref{lemma:lambda_I} we can directly generalise this results to any matrix (not necessarily Hermitian) $\rho_{AE}$. This implies that $\cT\circ\Theta$ is $\lambda$-expected contractive for any 
$\lambda\geq\sqrt{1-\frac{1}{\abs{A}^2}}\|\tau_{AB}\|_2$.
\end{proof}

\begin{lemma}
\label{lemma:lambda-multiplicativity}
Let $\Delta^{x_i}:A_i\rightarrow B_i$ be $\lambda_i$-expected-contractive maps, for $i\in I$, where $I$ is a finite index set and the $x_i$ are independent random variables. Then, the family
\[
  \Delta^{x_I}:A_I := \bigotimes_{i\in I} A_i \longrightarrow \bigotimes_{i\in I} B_i =: B_I,
\]
where $x_I=(x_i:i\in I)$, is $\lambda_I$-expected-contractive with $\lambda_I = \prod_{i\in I} \lambda_i$.
\end{lemma}
\begin{proof}
It is enough to prove the claim for $I=\{1,2\}$, as then the general case follows by induction on the cardinality of $I$.

Indeed, if $\Delta^{x_I}=\Delta^{x_A}\otimes\Delta^{x_B}$, then $\EE_{x_I} \|\Delta^{x_I}(\rho_{ABE})\|_2 = \EE_{x^{AB}} \|(\Delta^{x_A}\otimes\Delta^{x_B})(\rho_{ABE})\|_2$. If we define $\eta_{ABE}^{x_B} := (\1_A\otimes\Delta^{x_B})(\rho_{ABE})$ we can bound
\[\begin{split}
  \EE_{x^{AB}} \|(\Delta^{x_A}\otimes\Delta^{x_B})(\rho_{ABE})\|_2 
  &= \EE_{x^{AB}} \|\Delta^{x_A}(\eta_{ABE}^{x_B})\|_2 \\
  &\leq \lambda_A\EE_{x_B}\|\eta_{ABE}^{x_B}\|_2 \\
  &= \lambda_A\EE_{x_B}\|\Delta^{x_B}(\rho_{ABE})\|_2 \\
  &\leq \lambda_A\lambda_B\|\rho_{ABE}\|_2,
\end{split}\]
and we are done.
\end{proof}

We can join Lemmas \ref{lemma:lambda_I} and \ref{lemma:lambda-multiplicativity} in a single statement by making a distinction between the most general scenario where any general CP map $\cT_I:A_I\rightarrow B$ is applied, and the particular case where the map has a tensor product structure $\cT_I=\bigotimes_{i\in I}\cT_i$ such that $\cT_i:A_i\rightarrow B_i$, $B=\bigotimes_{i\in I} B_i$. In this second case, we can tighten the bound. We redact such a general statement in the following corollary.

\begin{corollary}
\label{corollary:D_I}
Given a CP map $\cT:A_I\rightarrow B$ with $I\subset \left[k\right]$, $\cT_I\circ\Theta_{A_I}$ is $\lambda_I$-expected-contractive with $\lambda_I=D_I\|\tau_{A_IB}\|_2$, where
\[
D_I = \begin{cases}
\frac{2^{\abs{I}-1}}{\prod_{i\in I}\sqrt{1-\frac{1}{\abs{A_i}^2}}} \phantom{=}\leq \frac12 \left(\frac{4}{\sqrt{3}}\right)^{|I|} &\text{for a general CP map $\cT_I:A_I\rightarrow B$}, \\
\prod_{i\in I}\sqrt{1-\frac{1}{\abs{A_i}^2}} \leq 1 &\text{when } \cT_I=\bigotimes_{i\in I}\cT_i \text{ with $\cT_i:A_i\rightarrow B_i$.}
\end{cases}
\] 
\end{corollary}
\begin{proof}
The first statement is actually Lemma \ref{lemma:lambda_I}, so it has already been proved. The second statement follows from Lemmas \ref{lemma:lambda.tensorial} and \ref{lemma:lambda-multiplicativity}. Notice that if the CP map has the commented tensor product structure, we can extract from Lemma \ref{lemma:lambda.tensorial} that $\cT_i\circ\Theta_{A_i}$ is $\lambda_i$-expected contractive with $\lambda_i=\sqrt{1-\frac{1}{\abs{A_i}^2}} \|\tau_{A_iB_i}\|_2$ for each system $A_i$. Now, from Lemma \ref{lemma:lambda-multiplicativity} we can calculate $\lambda_I=\prod_{i\in I}\lambda_i=\prod_{i\in I} \left(\sqrt{1-\frac{1}{\abs{A_i}^2}}\|\tau_{A_iB_i}\|_2\right) = \left(\prod_{i\in I}\sqrt{1-\frac{1}{\abs{A_i}^2}} \right) \|\tau_{A_IB}\|_2$.
\end{proof}

\begin{lemma}
\label{lemma:lambda_to_collision.entropy}
Consider a $\lambda$-randomizing family of channels $\cR^U:=\cT(U\cdot U^\dagger)$, where $\cT:A\rightarrow B$ is a CPTP map such that $\cT(\1_A/|A|)=\1_B/|B|$ with Choi operator $\tau_{AB}=(\id\otimes\cT)\Phi_{AA'}$, $U$ is distributed according to a probability law on $SU(A)$ that is a 2-design, and let us choose $\lambda^2=\Tr\tau_{AB}^2 \geq \left(1-\frac{1}{|A|}\right) \|\tau_{AB}\|_2^2$ from Lemma \ref{lemma:lambda.tensorial}. Then,
\[
\log\abs{B}+\log\lambda^2=-\widetilde{H}_2(A|B)_{\tau|\tau_B},
\]
where $\tau_B=\Tr_A \tau_{AB} = \1_B/\abs{B}$ is the maximally mixed state. 
Furthermore, for any $\cT_I: A_I\rightarrow B_I$, with $\cT_I\left(\frac{\mathds{1}_{A_I}}{|A_I|}\right)=\frac{\mathds{1}_{B_I}}{|B_I|}$, and $D_I$ given in Corollary \ref{corollary:D_I} we have
\[
\log\abs{B_I}+\log\lambda_I^2=2\log D_I-\widetilde{H}_2(A_I|B_I)_{\tau|\tau_{B}}.
\]
\end{lemma}
\begin{proof}
Applying the definition of the R\'enyi entropies we have
\[
-\widetilde{H}_2(A|B)_{\tau|\tau_B}=\log\Tr\left(\left[\left(\frac{\1_B}{\abs{B}}\right)^{-\frac{1}{4}}\!\!\tau_{AB}\!\left(\frac{\1_B}{\abs{B}}\right)^{-\frac{1}{4}}\right]^2\right)=\log\Tr(\abs{B}\tau_{AB}^2)=\log\abs{B}+\log\lambda'^2,
\]
where we have applied Lemma \ref{lemma:lambda.tensorial} in the last equality. Similarly, following Corollary \ref{corollary:D_I} for a CP map $\cT:A_I\rightarrow B$ we find
\[
\log\abs{B_I}+\log\lambda_I^2=2\log D_I+\log\Tr[\abs{B_I}\tau_{A_IB}^2]=2\log D_I-\widetilde{H}_2(A_I|B)_{\tau|\tau_{B_I}},
\]
concluding the proof.
\end{proof}

\section{Proving the multi-user decoupling theorems}
\label{sec:Proofs}

In Section \ref{sec:setting} we have found the bound
\[
\EE_{U_{[k]}} \left\| \cR^{U_{[k]}}(\rho_{A_{[k]}E}) - \tau_B\ox\rho_E \right\|_1
  \leq \sum_{\emptyset\neq I\subseteq[k]} \EE_{U_I} \left\| \left(\cT_I\circ\Theta_{A_I}\right)(\rho_{A_IE}) \right\|_1.
\]  
which allows us to treat each term of the sum on the right-hand side independently.

\begin{proofof}[of Theorem~\ref{theorem:with_smoothing}]
Let us define the modified objects $(\zeta_E^I)^{-\frac{1}{4}}\rho_{A_IE}(\zeta_E^I)^{-\frac{1}{4}}:=\Tilde{\rho}_{A_IE}$ and $(\sigma_B^I)^{-\frac{1}{4}}(\cT_I\circ\Theta_{A_I})(\cdot)(\sigma_B^I)^{-\frac{1}{4}}:=(\Tilde{\cT}_I\circ\Theta_{A_I})(\cdot)$, with a pair of states $\sigma_B^I$ and $\zeta_E^I$ chosen for each term $\emptyset\neq I\subseteq[k]$. Using Lemma \ref{lemma:tr_norm_bound} we can bound
\begin{equation}
\label{equation:change to tilde}
\begin{split}
\left\| (\cT_I\circ\Theta_{A_I})\rho_{A_IE} \right\|_1
&\leq\sqrt{
\Tr\left[\left((\sigma_B^I\otimes\zeta_E^I)^{-\frac{1}{4}}(\cT_I\circ\Theta_{A_I})\rho_{A_IE}(\sigma_B^I\otimes\zeta_E^I)^{-\frac{1}{4}} \right)^2\right]} \\
&=\sqrt{\Tr\left[\left((\Tilde{\cT}_I\circ\Theta_{A_I})\Tilde{\rho}_{A_IE}\right)^2\right] }
 =  \left\| (\Tilde{\cT}_I\circ\Theta_{A_I})\Tilde{\rho}_{A_IE} \right\|_2.
\end{split}
\end{equation}
Hence, the expected values are bounded as $\EE_{U_I}\!\left\| \left(\cT_I\circ\Theta_{A_I}\right)(\rho_{A_IE}) \right\|_1 \leq \EE_{U_I}\!\left\| (\Tilde{\cT}_I\circ\Theta_{A_I})\Tilde{\rho}_{A_IE} \right\|_2$. We extract from Corollary \ref{corollary:D_I} that $\tilde{\cT}_I\circ\Theta_{A_I}$ is $\lambda_I$-expected-contractive with $\lambda_I=D_I\|\tilde{\tau}_{A_IB}\|_2$. 
Therefore, $\EE_{U_I} \left\| (\Tilde{\cT}_I\circ\Theta_{A_I})\Tilde{\rho}_{A_IE} \right\|_2\leq D_I\left\|\tilde{\tau}_{A_IB}\right\|_2 \left\|\tilde{\rho}_{A_IE}\right\|_2$, with $D_I=2^{\abs{I}-1}\prod_{i\in I}\left(1-\frac{1}{\abs{A_i}^2}\right)^{-\frac{1}{2}}$ in the most general scenario. Now we unpack our tilde-modified operators to the original ones:
\[
D_I\|\tilde{\tau}_{A_IB}\|_2\|\tilde{\rho}_{A_IE}\|_2=D_I\|(\sigma_B^I)^{-\frac{1}{4}}\tau_{A_IB}(\sigma_B^I)^{-\frac{1}{4}}\|_2\|(\zeta_E^I)^{-\frac{1}{4}}\rho_{A_IE}(\zeta_E^I)^{-\frac{1}{4}}\|_2.
\]

Notice that we can always express sandwiched conditional R\'enyi entropies by means of Schatten-$\alpha$ norms as $2^{\frac{1-\alpha}{\alpha}\widetilde{H}_\alpha(A|B)_{\rho|\zeta}} = \left\|\zeta_E^{\frac{1-\alpha}{2\alpha}}\rho_{AE}\zeta_E^{\frac{1-\alpha}{2\alpha}}\right\|_\alpha$. Thus, 
\[
D_I\left\|(\sigma_B^I)^{-\frac{1}{4}}\tau_{A_IB}(\sigma_B^I)^{-\frac{1}{4}}\right\|_2 \left\|(\zeta_E^I)^{-\frac{1}{4}}\rho_{A_IE}(\zeta_E^I)^{-\frac{1}{4}}\right\|_2 = D_I2^{-\frac{1}{2}\widetilde{H}_2(A_I|E)_{\rho|\zeta_E^I}-\frac{1}{2}\widetilde{H}_2(A_I|B)_{\tau|\sigma_B^I}}.
\]

Now we can $\epsilon_I$-smooth each term $\emptyset\neq I\subseteq[k]$. That is, we consider states $\rho_{A_IE}'$ such that $\frac{1}{2}\|\rho_{A_IE}-\rho_{A_IE}'\|_1\leq\epsilon_I$. Thus, we can bound
\[
  \left\|(\cT\circ\Theta_{A_I})\rho_{A_IE}-(\cT\circ\Theta_{A_I})\rho_{A_IE}'\right\|_1 \leq 2^{\abs{I}}\cdot2\epsilon_I,
\]
because $\|\Theta_{A_i}\|_\diamond\leq2$ and so $\|\Theta_{A_I}\|_\diamond\leq2^\abs{I}$ due to the multiplicativity of the diamond norm under tensor products. Now using the triangle inequality we have
\begin{equation}\label{eq:colision_or_min}
\begin{split}
\EE_{U_I}\left\|(\cT\circ\Theta_{A_I})\rho_{A_IE}\right\|_1
 &\leq 2^{\abs{I}+1}\epsilon_I+\EE_{U_I} \left\|(\cT\circ\Theta_{A_I})\rho_{A_IE}'\right\|_1\\
 &\leq 2^{\abs{I}+1}\epsilon_I+D_I2^{-\frac{1}{2}\widetilde{H}_2(A_I|E)_{\rho'|\zeta_E^I}-\frac{1}{2}\widetilde{H}_2(A_I|B)_{\tau|\sigma_B^I}}\\
 &\leq 2^{\abs{I}+1}\epsilon_I+D_I2^{-\frac{1}{2}\widetilde{H}_{2}^{\epsilon_I}(A_I|E)_{\rho'|\zeta_E^I}-\frac{1}{2}\widetilde{H}_2(A_I|B)_{\tau|\sigma_B^I}},
\end{split}
\end{equation}
Where $\widetilde{H}_{2}^{\epsilon_I}$ is the smooth version of the sandwiched collision entropy, i.e. it is optimized over all possible states inside a $\epsilon_I$-ball around $\rho$ (a lower bound). This finally gives us:
\begin{equation}
\begin{split}
\EE_{U_{[k]}} &\left\| \cR^{U_{[k]}}(\rho_{A_{[k]}E}) - \tau_B\ox\rho_E \right\|_1 \\
&\leq \sum_{\emptyset\neq I\subseteq[k]} \left\{ 2^{|I|+1}\epsilon_I + D_I \exp_2\left[-\frac{1}{2}\widetilde{H}_2^{\epsilon_I}(A_I|E)_{\rho|\zeta_E^I}-\frac{1}{2}\widetilde{H}_2(A_I|B)_{\tau|\sigma_B^I}\right] \right\},   
\end{split}    
\end{equation}

Proving Theorem \ref{theorem:with_smoothing} and Corollary \ref{corollary:D_I-main} by changing the value of the constant $D_I$ according to the structure of the CP map as shown in Corollary \ref{corollary:D_I}.
\end{proofof}

Notice that the above bound can also be expressed in terms of min-entropies by lower-bounding  $\widetilde{H}_2(A_I|E)_{\rho'|\zeta_E^I}^{\epsilon_I}\geq H_\text{min}^{\epsilon_I}(A_I|E)_\rho$ (see Remark \ref{rem:2-vs-min}) on the last inequality of Equation~\eqref{eq:colision_or_min}.

\medskip
Dupuis \cite{Dupuis:Renyi-decouple} gave a bound on single-system decoupling using R\'enyi entropies; see also \cite{Mojahedian-et-al:correlation}. The main technical result in that paper states
\begin{equation}
\label{eq:Dupuis-interpolation}
 \EE_U\norm{\mathcal{N}^U(\rho_{AE})}_1\leq 2^{\frac{2}{\alpha}-1}\cdot2^{\frac{\alpha-1}{\alpha}(\log\abs{B}-\widetilde{H}_\alpha(A|E)_{\rho}+2\log\lambda)}
\end{equation}
for any family of CPTP maps $\mathcal{N}^U:A\rightarrow B$ that is a $\lambda$-expected contractive \cite[Lemma 7 \&{} Thm.~8]{Dupuis:Renyi-decouple}. Defining $\mathcal{N}^U:=\cR^U-\cD_A$, he finds
\begin{equation}
  \EE_U\norm{\cR^U(\rho_{AE})-\tau_B\ox\rho_E}_1\leq 2^{\frac{2}{\alpha}-1}\cdot2^{-\frac{\alpha-1}{\alpha}(\widetilde{H}_\alpha(A|E)_{\rho|\zeta_E}+\widetilde{H}_2(A|B)_{\tau|\tau_B})},  
\end{equation}
which we have rewritten using Lemma \ref{lemma:lambda_to_collision.entropy} in a more compact and recognisable form. The general case is a straightforward generalisation of this, as we have done all the heavy lifting before.

\begin{proofof}[of Theorem~\ref{theorem:without_smoothing}]
We start once again with
\[
  \EE_{U_{[k]}} \left\| \cR^{U_{[k]}}(\rho_{A_{[k]}E}) - \tau_B\ox\rho_E \right\|_1
  \leq \sum_{\emptyset\neq I\subseteq[k]} \EE_{U_I} \left\| \left(\cT_I\circ\Theta_{A_I}\right)\rho_{A_IE} \right\|_1.
\]
Just as in the previous proof, we can treat each term of the sum independently. Now, by defining $\mathcal{N}^U:=\cT_I\circ\Theta_{A_I}:A_I\rightarrow B_I$, we know from Lemma \ref{lemma:lambda_I} that this family of maps is $\lambda_I$-expected contractive. Thus, by Equation~\eqref{eq:Dupuis-interpolation},
\[
\EE_{U_I} \left\| \left(\cT_I\circ\Theta_{A_I}\right)\rho_{A_IE} \right\|_1\leq2^{\frac{2}{\alpha_I}-1}2^{\frac{\alpha_I-1}{\alpha_I}\left(\log|B_I|+2\log\lambda_I-\widetilde{H}_{\alpha_I}(A_I|E)_{\rho|\zeta_E^I}\right)},
\]
where $\alpha_I\in(1,2]$. Furthermore, from Lemma \ref{lemma:lambda_I} and more generally Corollary \ref{corollary:D_I} we know the value of $\lambda_I$, and actually we can identify $\log|B_I|+2\log\lambda_I=2\log D_I-\widetilde{H}_2(A_I|B)_{\tau|\tau_{B_I}}$ using Lemma \ref{lemma:lambda_to_collision.entropy}. Therefore we can finally write:
\begin{equation}
\begin{split}
\EE_{U_{[k]}} &\left\| \cR^{U_{[k]}}(\rho_{A_{[k]}E}) - \tau_B\ox\rho_E \right\|_1 \\
&\leq \sum_{\emptyset\neq I\subseteq[k]} 
          D_I^{2-\frac{2}{\alpha_I}}
          2^{\frac{2}{\alpha_I}-1}
          \exp\left[ \left(1-\frac{1}{\alpha_I}\right) \left(-\widetilde{H}_{\alpha_I}(A_I|E)_{\rho|\zeta^I_E} - \widetilde{H}_2(A_I|B_I)_{\tau|\tau_B} \right)\right],
\end{split} 
\end{equation}
concluding the proof of Theorem \ref{theorem:without_smoothing} and Corollary \ref{corollary:D_I-main}. 
\end{proofof}

\section{Applications}
\label{sec:apps}
To illustrate the power of our decoupling results, we shall discuss and solve four example problems in multi-user quantum information theory that have until now been hampered by the absence of the simultaneous smoothing technique. 
These are, in order: local randomness extraction from a given multipartite state in Subsection \ref{subsec:randomness}; concentration of multipartite pure entanglement in the hands of two designated users by LOCC, aka entanglement of assistance  in Subsection \ref{subsec:EoA}; quantum state merging, aka quantum Slepian-Wolf problem in Subsection \ref{subsec:Slepian-Wolf}; and finally quantum communication via quantum multiple access channels (MAC) in Subsection \ref{subsec:MAC}. 

For all of them, we first show how our decoupling bound yields a flexible one-shot achievability result, which in turn implies asymptotic rates in the i.i.d.~setting that in some cases had only been conjectured so far, or were known to rely on much more complicated proofs. We demonstrate furthermore the versatility of the one-shot bounds by generalising the i.i.d.~asymptotic rates to the case that the single-system state/channel is only partially known (compound source/channel setting).

In order to take this step from one-shot to i.i.d.~settings we make use of the quantum asymptotic equipartition property (AEP) \cite{Tomamichel:PhD}, which we state below along with a couple of other lemmas needed in the subsequent subsections. 

\begin{theorem}[AEP]
  \label{theorem:AEP}
  Let $\rho_{AB}$ be a bipartite state acting on $A\ox B$, so that for an integer $n$, $\rho_{AB}^{\ox n}$ is a state on $(A\ox B)^{\ox n}$. Then, for any $0<\epsilon<1$,
  \begin{align*}
   \lim_{n\rightarrow \infty}\frac{1}{n}H_{\min}^\epsilon(A^n|B^n)_{\rho^{\ox n}}
    &= S(A|B)_\rho, \\
  \phantom{================}
  \lim_{n\rightarrow \infty}\frac{1}{n}H_{\max}^\epsilon(A^n|B^n)_{\rho^{\ox n}} &=S(A|B)_\rho. 
  \phantom{===============:}
  \qedsymbol
\end{align*}
\end{theorem}

\begin{lemma}[{State space $\epsilon$-net~\cite{HLSW:randomizing}}]
  \label{lemma:net}
  For $\epsilon>0$ and an integer $d$, there exists a set $\cS_0$ of states on $\cS(\CC^d)$ with $M=|\cS_0| \leq \left(\frac{5}{\epsilon}\right)^{2d^2}$, 
  such that for every $\rho\in\cS(\CC^d)$ there exists a $\rho_0\in\cS_0$ with $\frac12\|\rho-\rho_0\|_1 \leq \epsilon$.
  \qed
\end{lemma}

\begin{lemma}[{Duality of R\'enyi entropies~\cite{Beigi:sandwich,Mueller-Lennert:sandwich}, see also \cite{TBH:Renyi}}]
  \label{lemma:Renyi-duality}  
  If $\alpha,\beta\in\left[\frac12,\infty\right]$ such that $\frac{1}{\alpha}+\frac{1}{\beta}=2$, then for any pure tripartite state $\psi_{ABC}$: $ \widetilde{H}_\alpha(A|B)_\psi = -\widetilde{H}_\beta(A|C)_\psi$. \qed
\end{lemma}

\begin{lemma}[{Classical conditioning~\cite[Prop.~9]{Mueller-Lennert:sandwich}}]
  \label{lemma:Renyi-cq-conditioning}
  For a cq-state $\rho_{ABY} = \sum_y \rho_{AB} \ox p_y\proj{y}_Y$ and any $\alpha>0$,
  \[
    \phantom{=================}
    \widetilde{H}_\alpha(A|BY)_{\rho} 
     \geq \min_y \widetilde{H}_\alpha(A|B)_{p_y}. 
    \phantom{================}
    \qedsymbol
  \]
\end{lemma}

\begin{lemma}
  \label{lemma:quasi-concavity}
  For any convex combination of $N$ states on $AB$, $\overline{\rho} = \sum_{i=1}^N p_i \rho_i$, and $0< \beta \leq \infty$,
  \[
    \widetilde{H}_\beta(A|B)_{\overline{\rho}} \leq \max_i \widetilde{H}_\beta(A|B)_{\rho_i} + \log N.
  \]
\end{lemma}
\begin{proof}
We show the bound only for $0<\beta<1$, for $\beta>1$ it is analogous, and for $\beta=1$ it follows from taking a limit (the case $\beta=\infty$ had been observed in \cite{MorganWinter:prettystrong}). 
Our starting point is the relation \cite[Prop.~2.9]{Mosonyi}
\[
  \sum_{i=1}^N p_i \widetilde{Q}_\beta(\rho_i\|\sigma) 
  \leq \widetilde{Q}_\beta(\overline{\rho}\|\sigma) 
  \leq \sum_{i=1}^N p_i^\beta \widetilde{Q}_\beta(\rho_i\|\sigma),
\]
for the sandwiched R\'enyi relative entropy and the quantity appearing inside the logarithm:
\begin{align*}
  \widetilde{D}_\beta(\rho\|\sigma) 
    = \frac{1}{\beta-1}\log \widetilde{Q}_\beta(\rho\|\sigma), \quad
  \widetilde{Q}_\beta(\rho\|\sigma)
    = \Tr\left[\sigma^{\frac{1-\beta}{2\beta}}\rho\sigma^{\frac{1-\beta}{2\beta}} \right]^\beta.
\end{align*}
We use the right-hand inequality and upper bound successively 
\[\begin{split}
  \widetilde{Q}_\beta(\overline{\rho}\|\sigma) 
  \leq \sum_{i=1}^N p_i^\beta \widetilde{Q}_\beta(\rho_i\|\sigma) 
  \leq \left(\max_i \widetilde{Q}_\beta(\rho_i\|\sigma)\right) \sum_{i=1}^N p_i^\beta 
  \leq \left(\max_i \widetilde{Q}_\beta(\rho_i\|\sigma)\right) N^{1-\beta},
\end{split}\]
the rightmost inequality by the concavity of the function $x^\beta$. Thus, 
\[
  \widetilde{D}_\beta(\overline{\rho}\|\sigma) \geq \min_i \widetilde{D}_\beta(\rho_i\|\sigma) - \log N,
\]
so finally for our conditional R\'enyi entropy, $\widetilde{H}_\beta(A|B)_\rho = \max_{\sigma_B} -\widetilde{D}_\beta\left(\rho_{AB}\|\1_A\ox\sigma_B\right)$,
\[\begin{split}
  \widetilde{H}_\beta(A|B)_{\overline{\rho}}
  &= - \widetilde{D}_\beta(\overline{\rho}_{AB}\|\1_A\ox\sigma_B) \\
  &\leq \max_i \left( -\widetilde{D}_\beta(\rho_i\|\1_A\ox\sigma_B)\right) + \log N \\
  &\leq \max_i \widetilde{H}_\beta(A|B)_{\rho_i} + \log N,
\end{split}\]
and we are done.
\end{proof}

\subsection{Local randomness extraction}
\label{subsec:randomness}
Randomness extraction aims to convert weak randomness into (almost) uniform random bits. If we hold some side information $E$ about the random variable $A$, we want our output to be perfectly random even with respect to the side information. That is to say, we want it not only to be uniform but also uncorrelated from $E$.

Measuring a state is a source of weak randomness, and each possible measure gives us a different probability distribution of the outcomes. We would like to bind the amount of randomness that can be extracted from an arbitrary state $\rho_A$ over all possible measurements. Even more, if we allow some side party $E$ to hold side quantum correlations, we want our output to be uniform and independent from it. This means that the processing $\cN_{A\rightarrow X}$ of the overall state should result in $\cN_{A\rightarrow X}(\rho_{AE})=\frac{\mathds1_{X}}{\abs{X}}\ox\rho_E$. From this, it is quite clear that there must be a connection between this problem and decoupling.

We want to go beyond this single-user scenario and study multipartite randomness extraction. This has been developed in \cite{YHW:randomness} in the i.i.d.~asymptotic setting for $k=2$. Here we consider a state $\rho_{A_1\ldots A_k E}$ of $k$ cooperating users $A_i$ and an eavesdropper $E$. The objective of the $A_i$ parties is to each make a destructive projective measurement $\{\Pi_{x_i}^{(i)}\}_{x_i\in [t_i]}:A_i\rightarrow X_i$ so that all random variables $X_i$ are jointly uniformly distributed and independent from $E$. We assume $t_i\leq |A_i|$ and identify the outcomes $x_i$ with basis states $\ket{x_i}$ of a $t_i$-dimensional Hilbert space $X_i$.
After the application of the POVM, we want the output state $\sigma_{X_{[k]} E}$ to satisfy
\begin{equation*}\begin{split}
  \sum_{x_1\in[t_1]}\dots\sum_{x_k\in[t_k]} &\proj{x_1}_{X_1}\otimes\cdots\otimes\proj{x_k}_{X_k}\otimes \Tr_{A_{[k]}} \left( \rho_{A_{[k]}E}\left( \Pi_{x_1}^{(1)}\otimes\cdots\otimes \Pi_{x_k}^{(k)}\otimes \1_E \right) \right) \\
  &\phantom{================}
   =:\sigma_{X_{[k]} E}\stackrel{!}{\approx} \frac{\1_{X_1}}{|X_1|}\otimes\cdots\otimes\frac{\1_{X_k}}{|X_k|}\otimes\rho_E.
\end{split}\end{equation*}
The trace distance
\[
  \delta := \frac12 \left\| \sigma_{X_{[k]} E} - \frac{\1_{X_1}}{|X_1|}\otimes\cdots\otimes\frac{\1_{X_k}}{|X_k|}\otimes\rho_E\right\|_1
\]
is called the \emph{error} of the protocol. 

In the base case $k=1$, this problem has been comprehensively studied in \cite{BertaFawziWehner:randomness}, where it was shown that $\log |X_1|$ can be as large as $\log|A_1| + (H_{\min}^\epsilon(A_1|E)_\rho)_- + O(\log\epsilon)$ while the error is $O(\epsilon)$, where $\left(x\right)_- = \min\{0,x\}$ is defined as the negative part of the real number $x$; furthermore, it was shown that this is essentially optimal. Looking at a subset $I\subseteq[k]$ of players and treating them as a single one, the optimality part of the result from \cite{BertaFawziWehner:randomness} shows that a protocol of error $\epsilon$ necessarily satisfies $\sum_{i\in I} \log|X_i| \leq \sum_{i\in I} \log|A_i| + (H_{\min}^\epsilon(A_I|E)_\rho)_-$ for all $I\subseteq[k]$. We will show that this can essentially be achieved. 

\begin{theorem}
  \label{theorem:local_rand_extr}
  Consider the randomness extraction setting above. 
  If the $t_i$ satisfy the following set of inequalities,
  \begin{equation}
  \label{eq:randomness-one-shot}
  \forall \emptyset\neq I\subseteq[k] 
  \quad 
  \sum_{i\in I} \log t_i \leq \sum_{i\in I} \log|A_i| + \left(H_{\min}^\epsilon(A_I|E)_\rho\right)_- + 2\log\epsilon,
  \end{equation}
  then there exists a one-shot local randomness extraction protocol with rates $\log t_i$ and error $\delta \leq (3^k+2^{k-1})\epsilon$
\end{theorem}
\begin{corollary}
  \label{corollary:local_rand_extr}
  Consider the i.i.d.~asymptotics of the state $\rho_{A_1\ldots A_k E}^{\ox n}$. The optimal rate region of the randomness rates $R_i=\frac{1}{n}\log t_i$ of bits per copy of the state in the limit of block length $n\rightarrow \infty$ and error $\delta\rightarrow 0$ is given by the set of inequalities
  \begin{equation}
    \label{eq:randomness-iid}
    \forall \emptyset\neq I\subseteq[k] \quad \sum_{i\in I} R_i \leq \sum_{i\in I} \log|A_i| + (S(A_I|E)_\rho)_-.
  \end{equation}
\end{corollary}
\begin{proof}
We prove here both Theorem \ref{theorem:local_rand_extr} and Corollary \ref{corollary:local_rand_extr}.
To achieve our goal, we let each party $i$ perform a random unitary $U_i$ on $A_i$ followed by a qc-channel $\cT_i(\alpha) = \sum_{x_i=1}^{t_i} \proj{x_i} \Tr\alpha P_{x_i}^{(i)}$ (which fulfills $\widetilde{H}_2(A_i|X_i)_{\tau_i} \geq \log \frac{|A_i|}{t_i}$), where $P_{x_i}^{(i)}$ are the projectors corresponding to each $x_i\in [t_i]$ possible outcome, therefore $\1_{A_i} = \sum_{x=1}^{t_i} P_x^{(i)}$. We impose an additional property on these projectors, they must have similar ranks. Actually, we do not let any pair of projectors differ in more than one unit in rank. This condition can be expressed as $\left\lfloor\frac{|A_i|}{t_i}\right\rfloor \leq \operatorname{rank} P_x^{(i)} \leq \left\lceil\frac{|A_i|}{t_i}\right\rceil$. For concreteness, let us sort them the greater first and the smaller after $\operatorname{rank} P_x^{(i)} = \left\lceil\frac{|A_i|}{t_i}\right\rceil$ for $x=1,\ldots,|A_i|\!\mod t_i$ and $\operatorname{rank} P_x^{(i)} = \left\lfloor\frac{|A_i|}{t_i}\right\rfloor$ for $x=(|A_i|\!\mod t_i) + 1,\ldots,t_i$. 

Now we can invoke Theorem \ref{theorem:with_smoothing} with Corollary \ref{corollary:D_I-main} (cf.~Corollary \ref{corollary:D_I}), finding that there exist unitaries $U_i$ on $A_i$ (found with high probability by sampling from a $2$-design) such that
\[\begin{split}
  \sigma_{X_1\ldots X_k E} 
  &= \left( \cT_1\!\circ\cU_1 \otimes\cdots\otimes \cT_k\circ\cU_k \otimes \id_E \right) \rho_{A_1\ldots A_kE} \\
  &=\sum_{x_1\in[t_1]}\dots\sum_{x_k\in[t_k]}\proj{x_1}_{X_1}\otimes\cdots\otimes\proj{x_k}_{X_K}\\
  &\phantom{==========}
   \otimes \Tr_{A_{[k]}} \left( \rho_{A_{[k]}E}\left( U_1^\dagger P_{x_1}^{(1)} U_1 \otimes\cdots\otimes U_k^\dagger P_{x_k}^{(k)} U_k \otimes \1_E \right) \right)
\end{split}\]
satisfies
\begin{equation}\begin{split}
  \label{eq:randomness-error}
  \frac12 &\left\| \sigma_{X_1\ldots X_k E} - \frac{\1_{X_1}}{|X_1|}\otimes\cdots\otimes\frac{\1_{X_k}}{|X_k|}\otimes\rho_E \right\|_1 \\
   &\phantom{=======}
  \leq \sum_{\emptyset\neq I\subseteq[k]} 2^{|I|}\epsilon + \frac12\sum_{\emptyset\neq I\subseteq[k]} \exp_2\left[-\frac{1}{2}\widetilde{H}_2^{\epsilon}(A_I|E)_{\rho|\zeta_E^I}-\frac{1}{2}\widetilde{H}_2(A_I|B)_{\tau|\sigma_B^I}\right]\\
  &\phantom{=======}
   \leq 3^k\epsilon 
       + \frac12 \sum_{\emptyset\neq I \subseteq [k]} \exp_2\left( -\frac12 \left( H_{\min}^\epsilon(A_I|E)_\rho + \log|A_I| - \sum_{i\in I} \log t_i \right) \right),
\end{split}\end{equation}
where we have chosen all $\epsilon_I=\epsilon$ to be equal and bounded $D_I\leq1$ in the first inequality, we have calculated $\sum_{\emptyset\neq I\subseteq[k]} 2^{|I|}\epsilon=(3^k-1)\epsilon\leq3^k\epsilon$, and bounded the conditional entropies using the arguments discussed at the beginning of the section. 
Finally, we can bound $ \delta \leq (3^k+2^{k-1})\epsilon$ if the following system of linear equations is satisfied: 
\[
  \forall \emptyset\neq I\subseteq[k] 
  \quad 
  \sum_{i\in I} \log t_i \leq \sum_{i\in I} \log|A_i| + H_{\min}^\epsilon(A_I|E)_\rho + 2\log\epsilon.\]
Since all $t_i\leq|A_i|$, the above inequality is trivially true unless $H_{\min}^\epsilon(A_I|E)_\rho$ is negative. So we might as well replace the min-entropy by its negative part $(H_{\min}^\epsilon(A_I|E)_\rho)_-$, resulting in the achievability of the region \eqref{eq:randomness-one-shot}. Together with the outer bound derived from \cite{BertaFawziWehner:randomness} this region is thus shown to be essentially optimal. This answers the question from \cite{YHW:randomness} about a one-shot version of the basic protocol and achievable rates from that paper, for all $k\geq 2$; and completes the proof of Theorem \ref{theorem:local_rand_extr}.

From this bound we also obtain easily a proof for Corollary \ref{corollary:local_rand_extr}. Namely, invoking the asymptotic equipartition property for the min-entropy (Theorem \ref{theorem:AEP}), a tuple of rates $R_i=\frac{1}{n}\log t_i\geq 0$ is achievable as $n\rightarrow\infty$ if and only if 
\[
  \forall \emptyset\neq I\subseteq[k] \quad \sum_{i\in I} R_i \leq \sum_{i\in I} \log|A_i| + (S(A_I|E)_\rho)_-,
\]
which concludes the proof, since the necessity of these bounds has been argued before \cite{YHW:randomness}.
\end{proof}

This reproduces the core result of \cite{YHW:randomness} for $k=2$, albeit with a much simpler protocol than there, and proves the conjectured rate region for all numbers $k$ of users.

To illustrate the benefit of being able to address each point in the achievable rate region directly, and via one-shot techniques, we consider the case that the i.i.d.~source state is only partially known, i.e.~it is $\rho_{A_{[k]}E}^{\ox n}$ with $\rho \in \cS \subset \cS(A_1\ldots A_kE)$. The objective in this so-called \emph{compound source} setting is to design a protocol that extracts randomness universally with the same figures of merit for all $\rho^{\ox n}$, $\rho \in\cS$. The following theorem also demonstrates the power of the R\'enyi entropic decoupling Theorem \ref{theorem:without_smoothing}.

\begin{theorem}
  \label{thm:compound-randomness}
  In the i.i.d.~limit of $n\rightarrow\infty$ and $\delta\rightarrow 0$, the achievable region of the rates $R_i = \frac{\log t_i}{n}$ for a compound source $\left( \rho^{\ox n} :  \rho\in\cS\right)$, is given by 
   \begin{equation}
    \label{eq:randomness-compound}
    \forall \emptyset\neq I\subseteq[k] \quad \sum_{i\in I} R_i \leq \sum_{i\in I} \log|A_i| + \inf_{\rho\in\cS} \,(S(A_I|E)_\rho)_-.
  \end{equation}
\end{theorem}
\begin{proof}
The optimality of the bounds follows from Corollary \ref{corollary:local_rand_extr}, since for a given subset $I\subseteq[k]$ and any $\rho\in\cS$ the bound $\sum_{i\in I} R_i \leq \sum_{i\in I} \log|A_i| + (S(A_I|E)_\rho)_-$ applies. It remains to prove the achievability. To this end, for block length $n$, we choose an $\frac{\eta}{n}$-net $\cS_0 \subset \cS$ to approximate elements of $\cS$ in trace norm. By adapting the proof of Lemma \ref{lemma:net}, we find $N:= |\cS_0| \leq \left(\frac{5n}{\eta}\right)^{2|A_{[k]}|^2|E|^2}$. We number the elements of the net, $\cS_0 = \{ \rho^{(y)} : y=1,\ldots,N \}$ and define the cq-state
\[
  \widetilde{\rho}_{A_{[k]}^nE^nY} 
    = \frac{1}{N} \sum_{y=1}^N  \rho_{A_{[k]}E}^{(y)\ox n} \ox \proj{y}_Y.
\]
The plan is to construct a protocol for this state, argue that hence it works well on each $\rho_{A_{[k]}E}^{(y)\ox n}$, and finally that it also works on every $\rho^{\ox n}$, $\rho\in \cS$ by the net property. We could do this directly using Theorem \ref{theorem:local_rand_extr}, except that we would have to make the smoothing parameter $\epsilon$ in the min-entropies dependent on $n$, which makes the argument awkward. Instead, we opt to use the R\'enyi decoupling from Theorem \ref{theorem:without_smoothing} (Corollary \ref{corollary:D_I-main}), following otherwise the proof of Theorem \ref{theorem:local_rand_extr}. This means that there, the error Equation~\eqref{eq:randomness-error} is modified to 
\begin{equation}\begin{split}
  \label{eq:randomness-error-compound}
  \frac12 &\left\| \sigma_{X_1\ldots X_k E^nY} - \frac{\1_{X_1}}{|X_1|}\otimes\cdots\otimes\frac{\1_{X_k}}{|X_k|}\otimes\widetilde{\rho}_{E^nY} \right\|_1 \\
  &\phantom{=======}
   \leq \sum_{\emptyset\neq I \subseteq [k]} \exp_2\left( -\frac{\alpha-1}{\alpha} \left( \widetilde{H}_\alpha(A_I^n|E^nY)_{\widetilde{\rho}} + n\log|A_I| - \sum_{i\in I} \log t_i \right) \right),
\end{split}\end{equation}
where we have chosen all $\alpha_I=\alpha>1$ equal. Now, the right-hand side of this bound is $\leq \delta$ if 
\[
  \forall\emptyset\neq I \subseteq[k] \quad \sum_{i\in I} \log t_i \leq \log|A_I^n| + \widetilde{H}_\alpha(A_I^n|E^nY)_{\widetilde{\rho}} + \frac{\alpha}{\alpha-1} \log\left(2^{-k}\delta\right), 
\]
for some $\delta\geq0$. However, we can lower-bound the conditional R\'enyi entropy here as follows:
\[\begin{split}
  \widetilde{H}_\alpha(A_I^n|E^nY)_{\widetilde{\rho}} 
  &\geq \min_y \widetilde{H}_\alpha(A_I^n|E^n)_{\rho^{(y)\ox n}} 
   = n\left( \min_y  \widetilde{H}_\alpha(A_I|E)_{\rho^{(y)}} \right) \\
  &\geq n\left( \inf_{\rho\in\cS}  \widetilde{H}_\alpha(A_I|E)_{\rho} \right) \\
  &\geq n\left(\inf_{\rho\in\cS} S(A_I|E)_{\rho} - \Delta(\alpha) \right),
\end{split}\]
where in the first line we have used Lemma \ref{lemma:Renyi-cq-conditioning} and the additivity of the conditional sandwiched R\'enyi entropy, in the second line that $\cS_0\subset\cS$, and finally in the third the uniform convergence of $\widetilde{H}_\alpha(A_I|E)$ to $S(A_I|E)$ as functions on state space. To explain the latter, $\widetilde{H}_\alpha(A_I|E)_\rho \rightarrow S(A_I|E)_\rho$ point-wise as $\alpha\rightarrow 1$, and all $\widetilde{H}_\alpha(A_I|E)$ and the limit $S(A_I|E)$ are continuous, hence uniformly continuous, functions on the compact state space. This implies that there exists $\Delta(\alpha)>0$ (converging to $0$ as $\alpha\rightarrow 1$) such that for all $I$ and all states $\rho$, $S(A_I|E)_\rho \geq \widetilde{H}_\alpha(A_I|E)_\rho \geq S(A_I|E)_\rho - \Delta(\alpha)$. With the rates $nR_i=\log t_i$, this implies that if 
\[
  \forall\emptyset\neq I \subseteq[k] \quad \sum_{i\in I} R_i \leq \log|A_I| + \inf_{\rho\in\cS} S(A_I|E)_\rho - \Delta(\alpha) + \frac1n \frac{\alpha}{\alpha-1} \log\left(2^{-k}\delta\right),
\]
then the right hand side of the bound \eqref{eq:randomness-error-compound} is $\leq \delta$. This means that the error of the same protocol on any one of the $\rho^{(y)\ox n}$ is $\leq N\delta$. Hence, by applying the triangle inequality twice, on any $\rho^{\ox n}$, $\rho\in\cS$, the error is $\leq N\delta+2\eta$. Letting $\delta = \frac{\eta}{N}$, the error ($\leq 3\eta$) can be made arbitrarily small, while the rates are bounded as
\[
  \forall\emptyset\neq I \subseteq[k] \quad \sum_{i\in I} R_i \leq \log|A_I| + \inf_{\rho\in\cS} S(A_I|E)_\rho - \Delta(\alpha) - O\left( \frac{\log n - \log\eta}{n(\alpha-1)} \right).
\]
For $n\rightarrow\infty$ and $\alpha\rightarrow 1$, this proves the claim.
\end{proof}

\subsection{Multi-party entanglement of assistance and assisted distillation}
\label{subsec:EoA}
Consider a pure state $\psi_{ABC_1\ldots C_m}$ of two parties $A$ and $B$ who are helped by $m$ other parties $C_i$ with the aim to obtain approximately a maximally entangled state $\Phi_d$ of Schmidt rank $d$ by using arbitrary local operations and classical communication (LOCC). Namely, if the overall CPTP map implemented by the LOCC protocol is denoted $\Lambda:ABC_1\ldots C_m \rightarrow A'B'$, with $|A'|=|B'|=d$, we aim to find 
\[
  \Lambda(\psi_{ABC_1\ldots C_m}) \stackrel{!}{\approx} (\Phi_d)_{A'B'},
\]
where $\ket{\Phi_d}$ is the standard maximally entangled state of Schmidt rank $d$. The trace distance 
\[
  \delta := \frac12 \left\| \Lambda(\psi_{ABC_1\ldots C_m}) - (\Phi_d)_{A'B'} \right\|_1
\]
is called the \emph{error} of the protocol. 

It is worth pausing for the simplest case, $m=0$, so that $\psi_{AB}$ is already a pure state between $A$ and $B$. Then the objective is merely to concentrate the entanglement by LOCC into maximal entanglement, and we find the essentially optimal $\log d \geq H_{\min}^\delta(\psi_A)$ \cite{LoPopescu}. For $m>0$, consider any bipartition of the helpers by choosing a subset $I\subseteq[m]$ and its complement $I^c=[m]\setminus I$, and simulate any $(m+2)$-party LOCC protocol by a bipartite LOCC protocol between the systems $AC_I$ and $BC_{I^c}$. Thus, from the preceding entanglement concentration considerations, we get the upper bound
\[
  \log d \leq \min_{\emptyset\subseteq I\subseteq[m]} H_{\min}^\delta(\psi_{AC_I}).
\]
We can show that this bound is essentially achievable, up to an additive offset depending only on $\delta$ and $m$, and a technical condition. 

\begin{theorem}
\label{theorem:EoA}
Given the setting above, multi-party entanglement of assistance has an achievable rate $\log d$ with error $\delta \leq 4\cdot 3^{m/2}\sqrt{\epsilon}$ if
\begin{equation}\begin{split}
  \label{eq:EoA-oneshot}
  \log d &\leq \min_{I\subseteq[m]} H_{\min}^\epsilon(AC_I)_\psi + 2\log\epsilon, \\
  -2\log\epsilon &\leq \min_{\emptyset\neq I\subseteq[m]} H_{\min}^\epsilon(C_I)_\psi.
\end{split}\end{equation}
\end{theorem}
\begin{corollary}
\label{corollary:EoA}
In the i.i.d.~limit of $n\rightarrow\infty$, the maximum asymptotic entanglement rate $R=\frac{1}{n}\log d$ from $\psi^{\ox n}$ is
\begin{equation}
  \label{eq:EoA-iid}
  R = \min_{I\subseteq[m]} S(\psi_{AC_I}),
\end{equation}
where $S(\rho)$ is the von Neumann entropy of the state $\rho$.
\end{corollary}

\begin{proof}
We prove both Theorem \ref{theorem:EoA} and Corollary \ref{corollary:EoA}.
To start with the former, our strategy will consist in making a local random complete basis measurement onto each $C_i$, and a random projective measurement of rank-$d$ projectors onto $A$; after that, $B$ will only have to perform a unitary. Let us fix orthonormal computational bases $\left\{\ket{j^{(i)}}\right\}$ for each $C_i$ with $i=1,\ldots,m$ and define a complete measurement in these bases as $\cT_i(\gamma) := \sum_{j^{(i)}=1}^{|C_i|} \proj{j^{(i)}} \gamma \proj{j^{(i)}}$. We also fix rank-$d$ projectors $P_{j^{(0)}}$ (we may assume w.l.o.g.~that $d$ divides the dimension $|A|$ by trivially enlarging $A$ if necessary), then $\cT_0(\alpha) = \sum_{j^{(0)}=1}^{|A|/d} P_{j^{(0)}} \alpha P_{j^{(0)}}$ is defined as the projective measurement of rank $d$ on $A$. Using that the R\'enyi entropies of the Choi states are $\widetilde{H}_2(C_i'|C_i)_{\tau_i} = 0$ ($\forall i\in[m]$) and $\widetilde{H}_2(A'|A)_{\tau_0} = - \log d$ \cite{Dupuis-et-al:decouple}, Theorem \ref{theorem:with_smoothing} with Corollary \ref{corollary:D_I-main} shows that there exist unitaries $U_0$ on $A$ and $U_i$ on $C_i$ ($i\in[m]$), found with high probability by sampling from a $2$-design, such that
\[
  \sigma_{AC_1\ldots C_m} = \left(\cT_0\circ\cU_0 \otimes \cT_1\circ\cU_1 \otimes \cdots \otimes \cT_m\circ\cU_m \right)\psi_{AC_1\ldots C_m}
\]
satisfies
\begin{equation}\begin{split}
  \label{eq:EoA-decoupling-bound}
  \frac12 &\left\| \sigma_{AC_1\ldots C_m} - \frac{\1_{A}}{|A|} \otimes \frac{\1_{C_1}}{|C_1|} \otimes \cdots \otimes \frac{\1_{C_m}}{|C_m|} \right\|_1 \\
  &\phantom{======}
   \leq 3^{m+1}\epsilon + \frac12 \sum_{\emptyset\subseteq I\subseteq[m]} 2^{-\frac12 (H_{\min}^\epsilon(AC_I)_\psi - \log d)} + \frac12 \sum_{\emptyset\neq I\subseteq[m]} 2^{-\frac12 H_{\min}^\epsilon(C_I)_\psi},
\end{split}\end{equation}
choosing all $\epsilon_I=\epsilon$ equal. The right hand side of this last bound is $\leq \eta := (3^{m+1}+2^m)\epsilon$ if the following conditions are satisfied: 
\begin{equation}\begin{split}
  \label{eq:EoA-oneshot_conditions}
  \log d &\leq \min_{I\subseteq[m]} H_{\min}^\epsilon(AC_I)_\psi + 2\log\epsilon, \\
  -2\log\epsilon &\leq \min_{\emptyset\neq I\subseteq[m]} H_{\min}^\epsilon(C_I)_\psi.
\end{split}\end{equation}
Let $\vec{j} = j^{(0)}j^{(1)}\ldots j^{(m)}$ be a set of possible measurement outcomes corresponding to the general POVM element $\Lambda_{\vec{j}}=P_{j^{(0)}}\ox\proj{j^{(1)}}\ox\cdots\ox\proj{j^{(m)}}$. The probability of getting this specific outcomes when measuring $\sigma_{AC_1\ldots C_m}$ is
\[\begin{split}
p(\vec{j})=\Tr\sigma_{AC_{[m]}}\Lambda_{\vec{j}} = \Tr \left[(U_0\ox U_1\ox\cdots\ox U_m)\psi_{AC_1\ldots C_m} (U_0\ox U_1\ox\cdots\ox U_m)^\dagger\Lambda_{\vec{j}}\right],
\end{split}\]
and the probability of obtaining the outcomes $\vec{j}$ after measuring the maximally mixed is given by
\[\begin{split}
  p'(\vec{j})=\Tr \left(\frac{\1_{A}}{|A|} \otimes \frac{\1_{C_1}}{|C_1|} \otimes \cdots \otimes \frac{\1_{C_m}}{|C_m|}\right)\Lambda_{\vec{j}}
  =\Tr\frac{P_{j^{(0)}}}{|A||C_1|\cdots|C_m|}
  =\frac{d}{|A||C_1|\cdots|C_m|}.
\end{split}\]
We can bound the total variational distance between the two probability distributions using Lemma \ref{lemma:tr_dist}:
\[\begin{split}
  \frac12 \sum_{\vec{j}}\abs{p(\vec{j})-p'(\vec{j})}
  &=\frac12\sum_{\vec{j}}\abs{\Tr\left(\sigma_{AC_1\ldots C_m}-\frac{\1_{A}}{|A|} \otimes \frac{\1_{C_1}}{|C_1|} \otimes \cdots \otimes \frac{\1_{C_m}}{|C_m|}\right)\Lambda_{\vec{j}}}\\
  &\leq \frac12 \left\| \sigma_{AC_1\ldots C_m} - \frac{\1_{A}}{|A|} \otimes \frac{\1_{C_1}}{|C_1|} \otimes \cdots \otimes \frac{\1_{C_m}}{|C_m|} \right\|_1 
  \leq\eta.
\end{split}\]
As $\sigma_{AC_{[m]}}$ and the maximally mixed state in the above trace distance are both direct sums over operators in the orthogonal subspaces given by the support of $\Lambda_{\vec{j}}$, we can rewrite the trace distance in question as
\[
  \frac12 \left\| \sigma_{AC_1\ldots C_m} - \frac{\1_{A}}{|A|} \otimes \frac{\1_{C_1}}{|C_1|} \otimes \cdots \otimes \frac{\1_{C_m}}{|C_m|} \right\|_1 
  = \sum_{\vec{j}} \frac12 \left\| \Lambda_{\vec{j}}\sigma_{AC_{[m]}}\Lambda_{\vec{j}} - p'(\vec{j})\Lambda_{\vec{j}} \right\|_1\leq\eta.
\]
Using the triangle inequality $\|\rho-\sigma\|_1\leq\|\rho-\tau\|_1+\|\tau-\sigma\|_1$ and the bound on the total variation distance between $p$ and $p'$ we can thus obtain
\begin{equation}\begin{split}
  \label{eq:EoA-norm-bound}
  \frac12 \sum_{\vec{j}} p(\vec{j}) &\left\| \frac{1}{p(\vec{j})} \left(\! P_{j^{(0)}}\ox\bra{j^{(1)}}\cdots\bra{j^{(m)}}\right) (U_0\ox U_1\ox\cdots\ox U_m) \right. \\
  &\phantom{===}\left.
   \psi_{AC_1\ldots C_m} (U_0\ox U_1\ox\cdots\ox U_m)^\dagger \left(\! P_{j^{(0)}}\ox\ket{j^{(1)}}\cdots\ket{j^{(m)}}\right) - \frac{P_{j^{(0)}}}{d} \right\|_1 \leq 2\eta.
\end{split}\end{equation}
Let us now introduce the unit vectors 
\[
  \ket{\psi(\vec{j})}_{AB} = \frac{1}{\sqrt{p(\vec{j})}} \left(\! P_{j^{(0)}}\ox\1_B\ox\bra{j^{(1)}}\cdots\bra{j^{(m)}}\right) (U_0\ox\1_B\ox U_1\ox\cdots\ox U_m) \ket{\psi}_{ABC_1\ldots C_m}, 
\]
so that we can define $\eta(\vec{j}) = \frac12\left\| \Tr_B \psi(\vec{j})_{AB} - \frac{P_{j^{(0)}}}{d} \right\|_1$, such that $\sum_{\vec{j}}  p(\vec{j})\eta(\vec{j}) \leq 2\eta$. We have a purification $\psi(\vec{j})_{AB}$, then by Uhlmann's Theorem \ref{theorem:Uhlmann}, there must exist a purification $\phi$ of the projector $\frac{P_j^{(0)}}{d}$ such that the purified distance is conserved. This is a maximally mixed state on its support, therefore any purification will be a maximally entangled state of rank $d$ (the dimension of the support) that we can write as $\ket{\Phi_d(\vec{j})}_{AB} = \left(U(\vec{j})\ox V(\vec{j})\right)\ket{\Phi_d}_{A'B'}$, where $U(\vec{j})$ and $V(\vec{j})$ are some isometries applied to the canonical maximally mixed state $\ket{\Phi_d}_{A'B'}$. Now, applying the Fuchs-van de Graaf inequalities \eqref{eq:FvdG_relation}, we find
\[
  \frac12 \left\|\psi(\vec{j})-\Phi_d(\vec{j})\right\|_1
   \leq P\left(\psi(\vec{j}),\Phi_d(\vec{j})\right)
   = P\left(\Tr_B\psi(\vec{j}),\frac{P_j^{(0)}}{d}\right)
   \leq\sqrt{\eta(\vec{j}) \left(2-\eta(\vec{j})\right)}. 
\]
With these elements and facts, we can finally describe the LOCC protocol to concentrate the entanglement in the hands of Alice and Bob: parties $A$ and the $C_i$ apply the local unitaries $U_0$ and $U_i$, followed by the projective measurements $(P_{j^{(0)}})$ and $\left(\proj{j^{(i)}}\right)$, respectively (in the case of the $C_i$ they are destructive). The measurement outcomes are broadcast to $A$ and $B$ who apply the (partial) isometries $U(\vec{j})^\dagger$ and $V(\vec{j})^\dagger$, respectively (see figure \ref{fig:EoAdiagram}). By the triangle inequality and the concavity of the square root, the resulting CPTP map $\Lambda:ABC_1\ldots C_m \rightarrow A'B'$ satisfies
\[
  \frac12 \left\| \Lambda(\psi_{ABC_1\ldots C_m}) - (\Phi_d)_{A'B'} \right\|_1 
    \leq \sqrt{2\eta(2-2\eta)} 
    \leq 2\sqrt{\eta}
    \leq 4\cdot 3^{m/2} \sqrt{\epsilon}.
\]
The achieved one-shot rate, always assuming that the second condition in \eqref{eq:EoA-oneshot_conditions} is fulfilled, is therefore $\log d \geq \min_{I\subseteq[m]} H_{\min}^\epsilon(AC_I)_\psi + 2\log\epsilon$. This concludes the proof of the theorem.

\begin{figure}[ht]
    \centering
    \includegraphics[scale=0.42]{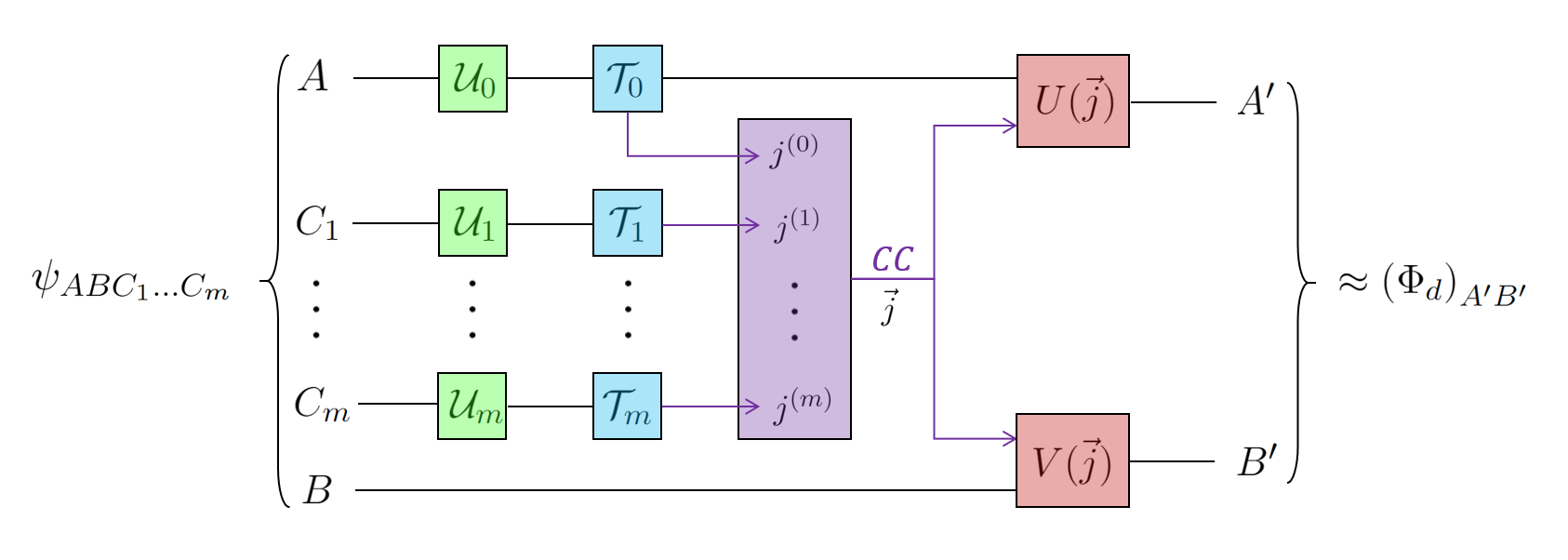}
    \caption{Diagram of the LOCC protocol that maximally concentrates the entanglement of an initial state $\psi_{ABC_1\dots C_m}$ onto Alice's and Bob's subspaces of dimension $|A'|=|B'|=d$.}
    \label{fig:EoAdiagram}
\end{figure}

To prove the corollary, referring to the i.i.d.~asymptotic limit of $n\rightarrow\infty$ copies of $\psi_{ABC_1\ldots C_m}$ and vanishing error $\delta\rightarrow 0$, the AEP applies. This means that $H_{\min}^\epsilon(A^nC_I^n)_{\psi^{\ox n}} \sim n S(\psi_{AC_I})$ and $H_{\min}^\epsilon(C_I^n)_{\psi^{\ox n}} \sim n S(\psi_{C_I})$. By the above comment, we may assume w.l.o.g. that all these von Neumann entropies are positive; for otherwise if some $S(\psi_{C_I})=0$ we can discard the corresponding parties, or if $S(\psi_{AC_I})=0$ then $A$ and $B$ are not entangled and there is nothing to distill by LOCC. In the positive case, all exponential terms in the sum for $\eta$ can be made exponentially or just sub-exponentially small in $n$, and defining the asymptotic rate via $\log d = nR$, we achieve its optimal value $R = \min_{I\subseteq[m]} S(\psi_{AC_I})$.
For the optimality, the necessity of the inequalities $R\leq S(AC_I)$ can be argued by noting that any LOCC protocol between the $m+2$ parties is at the same time an LOCC protocol for the bipartition $AC_I:BC_{[m]\setminus I}$, and the optimal rate for bipartite entanglement concentration is the reduced state entropy \cite{BBPS-e-concentration}.
\end{proof}

Corollary \ref{corollary:EoA} is the result from \cite{HOW:merging-CMP}, proved there by a much more complicated, iterative protocol that relied on the tensor product structure of $\psi^{\otimes n}$. The present procedure was previously analyzed by Dutil \cite[Ch.~5]{Dutil:PhD} and shown to work assuming the simultaneous smoothing conjecture in the i.i.d.~case. Here finally we achieve the same without any unproven conjectures. Note that in particular, no time sharing between different protocols is necessary.

The first of the conditions \eqref{eq:EoA-oneshot} is essentially necessary, making the achieved rate essentially optimal. The second condition looks like a technical artifact of the proof since we require that all the local measurement outcomes of the helpers $C_i$ are close to being uniformly distributed. However, this is not necessary for the objective of entanglement of assistance, but at the same time it becomes difficult to achieve by random basis measurements if some reduced state $\psi_{C_I}$ has rather small min-entropy. We can see that this is benign when $\psi_{C_I}$ is actually pure, as then our state factorizes, $\psi_{ABC_1\ldots C_m} = \psi_{ABC_{I^c}} \ox \psi_{C_I}$, and we can simply leave the parties $C_I$ out of the LOCC protocol without any loss. We have to leave the general case as an open question. In any case, we can observe that by providing a small amount of EPR states between any pair of players (in fact, the pairs $B$ and $C_j$ are sufficient), we can always ensure that the entropies $H_{\min}^\epsilon(C_I)$ are sufficiently lower bounded.

\begin{remark}
Generalising \cite{holography-networks}, Cheng \emph{et al.}~\cite{random-tensors} have considered a model of random multipartite pure state defined by starting with a multipartite pure states on a larger number or systems (original and auxiliary ones) and subjecting the auxiliary systems to local random measurements. 
This is more general than our objective of obtaining bipartite maximal entanglement, but in the bulk of the paper \cite{random-tensors} also more specific because there, the main interest is in initial states given by network of partially entangled bipartite pure states. Interestingly, to describe the resulting random states, the authors of \cite{random-tensors} manage to resolve the simultaneous smoothing in that special case. In the case of an arbitrary state, however, perhaps our current approach can help to gain insights into the high-probability properties of the resulting random states.
\end{remark}

\medskip
Looking at the proof of Theorem \ref{theorem:EoA}, we see that the attainability is essentially the same for an initial mixed state, as in the following theorem.

\begin{theorem}
\label{theorem:EoA-mixed}
Given a mixed state $\rho_{ABC_1\ldots C_m} = \Tr_E \proj{\psi}_{ABC_1\ldots C_mE}$ in the above setting, 
multi-party assisted distillable entanglement has an achievable rate $\log d$ with error $\delta \leq 4\cdot 3^{m/2}\sqrt{\epsilon}$ if
\begin{equation}\begin{split}
  \label{eq:EoA-mixed-oneshot}
  \log d &\leq \min_{I\subseteq[m]} -H_{\max}^\epsilon(AC_I|BC_{I^c})_\rho + 2\log\epsilon, \\
  -2\log\epsilon &\leq \min_{\emptyset\neq I\subseteq[m]} -H_{\max}^\epsilon(C_I|ABC_{I^c})_\rho.
\end{split}\end{equation}
\end{theorem}

\begin{proof}
We trace the proof of Theorem \ref{theorem:EoA}, indicating only the necessary changes. To start, we use the same random unitaries $U_j$ followed by the same $\cT_j$. Then Equation \eqref{eq:EoA-decoupling-bound} is replaced by 
\begin{equation}\begin{split}
  \label{eq:EoA-mixed-decoupling-bound}
  \frac12 &\left\| \sigma_{AC_1\ldots C_mE} - \frac{\1_{A}}{|A|} \otimes \frac{\1_{C_1}}{|C_1|} \otimes \cdots \otimes \frac{\1_{C_m}}{|C_m|} \otimes \psi_E \right\|_1 \\
  &\phantom{======}
   \leq 3^{m+1}\epsilon + \frac12 \sum_{\emptyset\subseteq I\subseteq[m]} 2^{-\frac12 (H_{\min}^\epsilon(AC_I|E)_\psi - \log d)} + \frac12 \sum_{\emptyset\neq I\subseteq[m]} 2^{-\frac12 H_{\min}^\epsilon(C_I|E)_\psi},
\end{split}\end{equation}
with respect to the purification $\psi$ of $\rho$ on the right hand side and for 
\[
  \sigma_{AC_1\ldots C_mE} = \left(\cT_0\circ\cU_0 \otimes \cT_1\circ\cU_1 \otimes \cdots \otimes \cT_m\circ\cU_m \otimes \id_E\right)\rho_{AC_1\ldots C_mE}.
\]

Then, if the conditions \eqref{eq:EoA-mixed-oneshot} are satisfied, the right hand side of Equation \eqref{eq:EoA-mixed-decoupling-bound} becomes $\leq \left(2^m+3^{m+1}\right)\epsilon=: \eta$, and we can continue as before until Equation \eqref{eq:EoA-norm-bound}, which is replaced by 
\begin{equation}\begin{split}
  \label{eq:EoA-mixed-norm-bound}
  \frac12 \sum_{\vec{j}} p(\vec{j}) &\left\| \frac{1}{p(\vec{j})} \left(\! P_{j^{(0)}}\ox\bra{j^{(1)}}\cdots\bra{j^{(m)}}\right) (U_0\ox U_1\ox\cdots\ox U_m) \right. \\
  &\phantom{===}\left.
   \psi_{AC_1\ldots C_mE} (U_0\ox U_1\ox\cdots\ox U_m)^\dagger \left(\! P_{j^{(0)}}\ox\ket{j^{(1)}}\cdots\ket{j^{(m)}}\right) - \frac{P_{j^{(0)}}}{d} \otimes \psi_E \right\|_1 \leq 2\eta.
\end{split}\end{equation}

The rest of the proof is almost unchanged, except that the purification $\ket{\phi}_{\widetilde{E}E}$ of $\psi_E$ comes in: with 
\[
  \ket{\psi(\vec{j})}_{ABE} \! = \!\frac{1}{\sqrt{p(\vec{j})}} \left(\! P_{j^{(0)}}\ox\1_{BE}\ox\bra{j^{(1)}}\cdots\bra{j^{(m)}}\right) (U_0\ox\1_{BE}\ox U_1\ox\cdots\ox U_m) \ket{\psi}_{ABEC_1\ldots C_m},
\]
we have $\eta(\vec{j}) = \frac12\left\| \Tr_B \psi(\vec{j})_{ABE} - \frac{P_{j^{(0)}}}{d}\otimes\psi_E \right\|_1$, such that $\sum_{\vec{j}}  p(\vec{j})\eta(\vec{j}) \leq 2\eta$. Once again, through Uhlmann's theorem, gives us isometries $U(\vec{j}):A'\hookrightarrow A$ and $V(\vec{j}):B'\widetilde{E} \hookrightarrow B$ such that 
\[
  \left(U(\vec{j})\ox V(\vec{j})\right)\left(\ket{\Phi_d}_{A'B'}\otimes\ket{\phi}_{\widetilde{E}E}\right) \approx \ket{\psi(\vec{j})},
\]
and the proof concludes exactly as before. Finally note that $H_{\min}^\epsilon(AC_I|E) = -H_{\max}^\epsilon(AC_I|BC_{I^c})$ by the duality relation between min- and max-entropies (cf.~\cite{Tomamichel:PhD}). 
\end{proof}

Unlike the pure-state case we do not have any clear statement of optimality of the rate achieved in this theorem. In fact, due to properties of the coherent information one should optimise the expressions in \eqref{eq:EoA-mixed-oneshot} over preprocessing channels $\cT_j:C_j\rightarrow C_j'$ ($j=1,\ldots,m$) and $cT_0:A\rightarrow A'$, and also consider swapping the roles of $A$ and $B$. Even then, we are limited by the specific protocol we are considering (rather than a general LOCC procedure); furthermore, in the i.i.d.~asymptotic limit regularisation might be required. Nevertheless we can state the following result. 

\begin{corollary}[{Cf.~Dutil~\cite[Thm.~5.4.4]{Dutil:PhD}}]
Given asymptotically many copies of a state $\rho_{ABC_{[m]}}$ ($n\gg 1$) and $o(n)$ EPR states between any pair of players, the following rate is achievable for EPR distillation between $A$ and $B$ by LOCC assisted by the players $C_j$: 
\begin{equation}
  \label{eq:EoA-mixed-rate}
  R = \sup_{\cT_j} \min_{I\subseteq[m]} I(A'C_I'\rangle BC_{I^c}')_\sigma \text{ s.t. } \forall I\subseteq[m]\ I(C_I'\rangle ABC_{I^c}')_\sigma \geq 0,
\end{equation}
where the supremum is over channels $\cT_j:C_j\rightarrow C_j'$ ($j=1,\ldots,m$) and $\cT_0:A\rightarrow A'$, and 
\[
  \phantom{============}
  \sigma_{A'BC_1'\ldots C_m'} = (\cT_0 \ox \id_B \ox \cT_1 \ox\cdots\ox \cT_m)\rho_{ABC_1\ldots C_m}.
  \phantom{============}
  \qedsymbol
\]
\end{corollary}

To demonstrate a case where the one-shot result are relevant, we consider the problem of i.i.d.~entanglement of assistance when the source is only partially known, meaning $\psi \in \cS \subseteq \cS(ABC_{[m]})$, and we would like to design protocols as above for every $n$ that are universal for all $\ket{\psi}^{\ox n}\in A^nB^nC_{[m]}^n$ with $\psi\in\cS$ (i.e. a compound source). 

\begin{theorem}
\label{thm:compound-EoA}
In the i.i.d.~limit of $n\rightarrow\infty$ and error $\delta\rightarrow 0$, the maximum entanglement rate $R=\frac1n\log d$ for a compound source $(\psi^{\ox n}:\psi\in\cS)$, when $o(n)$ EPR states are available for free between any pair or players, is
\begin{equation}
  \label{eq:EoA-compound-iid-rate}
  R  = \inf_{\psi\in\cS} \min_{I\subseteq[m]} S(AC_I)_\psi.
\end{equation}
\end{theorem}

\begin{proof}
Even when the state $\psi^{\ox}$ is fixed and known, Corollary \ref{corollary:EoA} upper-bounds the rate by $\min_{I\subseteq[m]} S(AC_I)_\psi$, hence the infimum of this quantity over $\psi\in\cS$ provides an upper bound on the optimal rate $R$.

Regarding the achievability, for block length $n$ choose an $\frac{\eta}{n}$-net $\cS_0\subset\cS$ of states for $\cS$ (i.e.~a net to approximate elements of $\cS$). By adapting the proof of Lemma \ref{lemma:net}, we find that  $N := |\cS_0| \leq \left(\frac{5n}{\eta}\right)^{2|A||B||C_{[m]}|}$. We number the elements of the net, $\cS_0 = \{\rho_s : s=1,\ldots,N\}$. The plan is to construct a one-shot assisted entanglement distillation protocol for the averaged state plus sublinear entanglement, 
\begin{equation*}
\widetilde{\rho}_{A^nB^nC_{[m]}^n} \ox \widetilde{\Phi}_{b_{[m]}c_{[m]}}
  := \left(\frac{1}{N} \sum_{s=1}^N \psi_s^{\ox n} \right) \ox \left(\bigotimes_{j=1}^m \Phi_{b_jc_j}\right)
  \text{ on } \widetilde{A}\widetilde{B}\widetilde{C}_{[m]} = A^nB^nb_{[m]}C_{[m]}^nc_{[m]},
\end{equation*}
where $\Phi_{b_jc_j}$ is a maximally entangled state between systems $b_j$ (with Bob) and $c_j$ (with helper $j$) of dimensions $|b_j|=|c_j|=d_0$, $\log d_0=o(n)$, $\widetilde{A}=A^n$, $\widetilde{B}=B^nb_{[m]}$ and $\widetilde{C}_j=C_j^nc_j$. 
Then to argue that the protocol performs well on all $\psi_s^{\ox n}$ (plus the sublinear entanglement), and finally that it must perform well on all $\psi^{\ox n}$ with $\psi \in \cS$, we could do this directly using Theorem \ref{theorem:EoA-mixed}, except that for that to work we have to make the smoothing parameter $\epsilon$ in the min-entropies dependent on $n$, which makes the argument awkward. Instead, we opt to use the R\'enyi decoupling from Theorem \ref{theorem:without_smoothing} (Corollary \ref{corollary:D_I-main}), following otherwise the proof of Theorem \ref{theorem:EoA-mixed}. This means that there, Equation \eqref{eq:EoA-mixed-decoupling-bound} is replaced by 
\begin{equation}\begin{split}
  \label{eq:EoA-compound-decoupling-bound}
  \frac12 &\left\| \widetilde{\sigma}_{A^n\widetilde{C}_1\ldots \widetilde{C}_mE} - \frac{\1_{A}^{\ox n}}{|A|^n} \otimes \frac{\1_{\widetilde{C}_1}}{|\widetilde{C}_1|} \otimes \cdots \otimes \frac{\1_{\widetilde{C}_m}}{|\widetilde{C}_m|} \otimes \widetilde{\psi}_E \right\|_1 \\
  &\phantom{=}
   \leq \sum_{\emptyset\subseteq I\subseteq[m]} \!\!\!\exp_2\left( -\frac{\alpha-1}{\alpha} (\widetilde{H}_{\alpha}(A^n\widetilde{C}_I|E)_{\widetilde{\psi}\ox\widetilde{\Phi}} - \log d) \right) + \sum_{\emptyset\neq I\subseteq[m]} \!\!\!\exp_2\left(-\frac{\alpha-1}{\alpha} \widetilde{H}_{\alpha}(\widetilde{C}_I|E)_{\widetilde{\psi}\ox\widetilde{\Phi}} \right)\!,
\end{split}\end{equation}
with respect to the purification $\widetilde{\psi}$ of $\widetilde{\rho}$ and for 
\[
  \widetilde{\sigma}_{A^n\widetilde{C}_1\ldots \widetilde{C}_mE} = \left(\cT_0\circ\cU_0 \otimes \cT_1\circ\cU_1 \otimes \cdots \otimes \cT_m\circ\cU_m \otimes \id_E\right) (\widetilde{\psi}\ox\widetilde{\Phi}),
\]
with the maps $\cU_0$ and $\cT_0$ acting on $A^n$, and $\cU_j$ and $\cT_j$ acting on $\widetilde{C}_j$. 

We upper-bound the right hand side of Equation \eqref{eq:EoA-compound-decoupling-bound} as follows: with $\frac{1}{\alpha}+\frac{1}{\beta} = 2$,
\[\begin{split}
  -\widetilde{H}_{\alpha}(A^n\widetilde{C}_I|E)_{\widetilde{\psi}\ox\widetilde{\Phi}}
  &\leq -\widetilde{H}_{\alpha}(A^nC_I^n|E)_{\widetilde{\psi}} - |I|\log d_0 \\
  &\leq -\widetilde{H}_{\alpha}(A^nC_I^n|E)_{\widetilde{\psi}} \\
  &=     \widetilde{H}_{\beta}(A^nC_I^n|B^nC_{I^c}^n)_{\widetilde{\psi}} \\
  &\leq \max_{s\in[N]} \widetilde{H}_{\beta}(A^nC_I^n|B^nC_{I^c}^n)_{\psi_s^{\ox n}} + \log N \\
  &\leq \sup_{\psi\in\cS} \widetilde{H}_{\beta}(A^nC_I^n|B^nC_{I^c}^n)_{\psi^{\ox n}} + \log N \\
  &=   -n \inf _{\psi\in\cS} \widetilde{H}_{\alpha}(AC_I)_{\psi} + \log N,
\end{split}\]
and similarly, 
\[\begin{split}
  -\widetilde{H}_{\alpha}(\widetilde{C}_I|E)_{\widetilde{\psi}\ox\widetilde{\Phi}}
  &\leq -\widetilde{H}_{\alpha}(C_I^n|E)_{\widetilde{\psi}} - |I|\log d_0 \\
  &\leq -\widetilde{H}_{\alpha}(C_I^n|E)_{\widetilde{\psi}} - \log d_0 \\
  &=     \widetilde{H}_{\beta}(C_I^n|A^nB^nC_{I^c}^n)_{\widetilde{\psi}} - \log d_0 \\
  &\leq \max_{s\in[N]} \widetilde{H}_{\beta}(C_I^n|A^nB^nC_{I^c}^n)_{\psi_s^{\ox n}} + \log N - \log d_0 \\
  &\leq \sup_{\psi\in\cS} \widetilde{H}_{\beta}(C_I^n|A^nB^nC_{I^c}^n)_{\psi^{\ox n}} + \log N - \log d_0 \\
  &=   -n \inf _{\psi\in\cS} \widetilde{H}_{\alpha}(C_I)_{\psi} + \log N - \log d_0,
\end{split}\]
in both chains of inequalities using Lemmas \ref{lemma:Renyi-duality} (for the equalities) and \ref{lemma:quasi-concavity} (for the inequalities in the fourth line) and the additivity of the conditional R\'enyi entropy, and in the second chain additionally that $I\neq\emptyset$.

Thus, with the rate $nR = \log d$, the right hand side of the bound \eqref{eq:EoA-compound-decoupling-bound} is $\leq\delta$ if 
\[\begin{split}
  R        &\leq \inf_{\psi\in\cS} \min_{I\subseteq[m]} \widetilde{H}_\alpha(AC_I)_\psi - \frac1n\left(\log N - \frac{\alpha}{\alpha-1}\log\left(2^{-m-1}\delta\right) \right), \\
  \log d_0 &\geq \log N - \frac{\alpha}{\alpha-1}\log\left(2^{-m-1}\delta\right).
\end{split}\]
Since $\widetilde{H}_\alpha(AC_I)_{\rho}$ converges to $S(AC_I)_{\rho}$ as $\alpha\rightarrow 1$, and the converging as well as the limit functions are continuous on the compact set of all states, hence uniformly continuous, also the convergence  $\widetilde{H}_\alpha(AC_I) \rightarrow S(AC_I)$ of the functions on state space is uniform. Thus, there exists a $\Delta(\alpha)>0$ (converging to $0$ as $\alpha\rightarrow 1$) such that for all $I\subseteq[m]$, 
\[
  \inf_{\psi\in\cS} S(AC_I)_{\psi} 
    \geq \sup_{\psi\in\cS} \widetilde{H}_\alpha(AC_I)_{\psi} 
    \geq \sup_{\psi\in\cS} S(AC_I)_{\psi} - \Delta(\alpha).
\]
And so the trace norm in Equation~\eqref{eq:EoA-compound-decoupling-bound} is guaranteed to be $\leq\delta$ if
\[\begin{split}
  R        &\leq \inf_{\psi\in\cS} \min_{I\subseteq[m]} S(AC_I)_\psi - \Delta(\alpha) - \frac1n\left(\log N - \frac{\alpha}{\alpha-1}\log\left(2^{-m-1}\delta\right) \right), \\
  \log d_0 &\geq \log N - \frac{\alpha}{\alpha-1}\log\left(2^{-m-1}\delta\right).
\end{split}\]
Continuing the reasoning of the proof of Theorem \ref{theorem:EoA-mixed}, we obtain an assisted distillation protocol for $\widetilde{\rho}$ that has error $\leq 2\sqrt{\delta}$, hence it has error $\leq 2N\sqrt{\delta}$ on each of the $\psi_s^{\ox n}$, and so finally it has error $\leq 2N\sqrt{\delta}+\eta$ on each source $\psi^{\ox n}$ such that $\psi\in\cS$. Choosing $\delta=\frac{\eta}{N^2}$ and $\log d_0 = \frac{3\alpha}{\alpha-1}(\log N - \log\eta +m+1)$, we get an error guarantee of $\leq 3\eta$ across the set $\cS$, while the rate achieved is 
\[
  R = \inf_{\psi\in\cS} \min_{I\subseteq[m]} S(AC_I)_\psi - \Delta(\alpha) - O\left( \frac{\log n-\log\eta}{n(\alpha-1)} \right),
\]
which for $n\rightarrow\infty$ and $\alpha\rightarrow 1$ proves the claim. 
\end{proof}

\subsection{Multi-party quantum Slepian-Wolf coding: state merging}
\label{subsec:Slepian-Wolf}
In terms of decoupling strategy and objectives, this task could be considered a generalisation of the previous, entanglement of assistance, except that we are interested in both entanglement yield and entanglement consumption and their net difference. Namely, the setting is described by a pure state $\psi_{A_1\ldots A_kBR}$ of $k+2$ parties, $k$ senders (Alice-$i$) holding $A_i$, one receiver (Bob) holding $B$ and a reference system $R$, whose only role is to hold the purification. Additionally the parties share maximally entangled states $\Phi_{A_i'B_i'}$ between Alice-$i$ and Bob of Schmidt rank $c_i$, so that the overall initial state is 
\[
  \psi_{A_1\ldots A_kBR} \ox (\Phi_{c_1})_{A_1'B_1'} \ox \cdots \ox (\Phi_{c_k})_{A_k'B_k'}.
\]

A one-way LOCC state merging protocol consists first of $k$ compression (encoding) instruments $\left( \cE_i^{(x)}:A_iA_i' \rightarrow A_i'' : x\in[\ell_i]\right)$, with the individual maps acting as $\cE_i^{(x)}(\alpha) = V_i^{(x)\dagger} \alpha V_i^{(x)}$. Here, the $V_i^{(x)}:A_i''\rightarrow A_iA_i'$ are isometries, i.e. $V_i^{(x)\dagger}V_i^{(x)}=\1_{A_i''}$, such that 
the projectors $\Pi_i^{(x)} = V_i^{(x)}V_i^{(x)\dagger}$ form a projective measurement, i.e. $\sum_{x=1}^{\ell_i} \Pi_i^{(x)} = \1_{A_iA_i'}$. We denote $|A_i''|=d_i$, $|X_i|=\ell_i$, hence $|A_i|c_i=d_i\ell_i$, which might necessitate to increase $A_i$ by isometric embedding. 
Secondly, of a collection of decompression (decoding) CPTP maps $\cD^{(x_{[k]})}:BB_1'\ldots B_k' \rightarrow \widehat{A}_1\ldots\widehat{A}_k\widehat{B}B_1''\ldots B_k''$, one for each tuple $x_{[k]}=x_1\ldots x_k$ of outcomes (where $\hat{A}_i\cong A_i$ and $\hat{B}\cong B$). 
The idea is that Alice-$i$ performs the instrument $\cE_i$, obtaining outcome $x_i$ which is communicated to Bob, who collects the outcome tuple $x_{[k]}$ and applies $\cD^{x_{[k]}}$. The result is a one-way LOCC operation $\Lambda:A_1\ldots A_k A_1'\ldots A_k' BB_1'\ldots B_k' \rightarrow \widehat{A}_1\ldots\widehat{A}_k\widehat{B} A_1''\ldots A_k'' B_1''\ldots B_k''$ that can be written as
\[
  \Lambda = \sum_{x_{[k]}} \cE_1^{(x_1)} \ox \cdots \ox \cE_k^{(x_k)} \ox \cD^{(x_{[k]})}.
\]
The objective is that at the end, after application of $\Lambda$, the Alices and Bob share approximately the state $\psi_{\widehat{A}_1\ldots \widehat{A}_k\widehat{B}R} \ox (\Phi_{d_1})_{A_1'B_1'} \ox \cdots \ox (\Phi_{d_k})_{A_k'B_k'}$, where now $\widehat{A}_1\ldots \widehat{A}_k\widehat{B}$ are held by Bob, and Alice-$i$ shares with Bob maximally entangled state of Schmidt rank $d_i$: 
\[
  \Lambda\left(\psi_{A_1\ldots A_kBR} \ox (\Phi_{c_1})_{A_1'B_1'} \ox \cdots \ox (\Phi_{c_k})_{A_k'B_k'}\right) \stackrel{!}{\approx} 
  \psi_{\widehat{A}_1\ldots \widehat{A}_k\widehat{B}R} \ox (\Phi_{d_1})_{A_1''B_1''} \ox \cdots \ox (\Phi_{d_k})_{A_k''B_k''}.
\]
The trace distance 
\[
  \eta := \frac12\! \left\| \Lambda\!\!\left(\psi_{A_1\ldots A_kBR} \!\ox\! (\Phi_{c_1})_{A_1'B_1'} \!\ox \cdots \ox\! (\Phi_{c_k})_{A_k'B_k'}\right) \!- 
  \psi_{\widehat{A}_1\ldots \widehat{A}_k\widehat{B}R} \!\ox\! (\Phi_{d_1})_{A_1''B_1''} \!\ox \cdots \ox\! (\Phi_{d_k})_{A_k''B_k''} \right\|_1
\]
is called the \emph{error} of the protocol.

Let us define the numbers $r_i := \log c_i - \log d_i$ as the net one-shot rates of entanglement cost for Alice-$i$, and the task is to characterize the possible tuples of these rates with corresponding state merging protocols. This problem has been introduced and solved in \cite{HOW:merging-Nature,HOW:merging-CMP} in the asymptotic setting of both single and multiple senders, and in \cite{Berta:merging} in the one-shot setting of a single sender. Dutil \cite{Dutil:PhD} has investigated the case of multiple senders in the one-shot setting as well as in the i.i.d.~asymptotics, and made the connection to the question of simultaneous smoothing of collision entropies and min-entropies \cite{DutilHayden:merging}. 

\begin{theorem}
\label{theorem:merging}
Given the setting above, quantum state merging can be achieved with error $\eta \leq 4\cdot 3^{k/2}\sqrt{\epsilon}$ if  
\begin{equation}\begin{split}
  \label{eq:one-shot-merging}
  \forall \emptyset\neq I\subseteq[k] 
  \quad 
  \sum_{i\in I} \log d_i &\leq H_{\min}^\epsilon(A_I|R)_\psi + \sum_{i\in I} \log c_i + 2\log\epsilon, \\
  \text{or equivalently } \sum_{i\in I} r_i &\geq H_{\max}^\epsilon(A_I|A_{I^c}B)_\psi - 2\log\epsilon,
\end{split}\end{equation}
with the above net one-shot rates of entanglement consumption $r_i=\log c_i-\log d_i$.

\end{theorem}
\vspace{-5pt}
\begin{figure}[ht]
    \centering
    \includegraphics[scale=0.6]{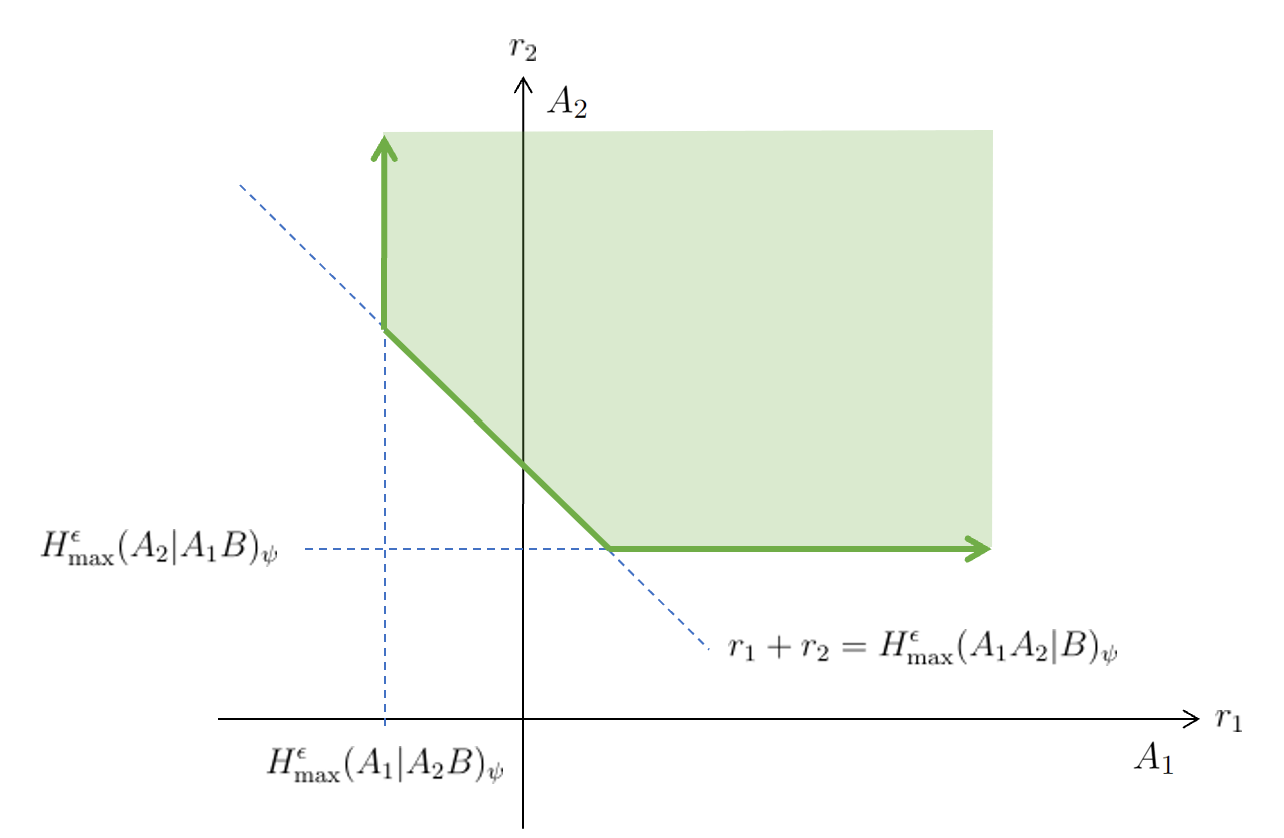}
    \caption{One-shot achievable rate region of a two-senders quantum Slepian-Wolf coding. Notice that the region is open towards the northeast.}
    \label{fig:RR_state_merging}
\end{figure}

\begin{corollary}
\label{corollary:merging}
In the i.i.d.~limit of $n\rightarrow\infty$, the region of achievable rates $R_i=\frac1n r_i$ for successful quantum state merging of $\psi^{\ox n}$ is given precisely by
\begin{equation}
  \label{eq:q-SW-iid}
  \forall I\subseteq[k] \quad \sum_{i\in I} R_i \geq S(A_I|A_{I^c}B)_\psi. 
\end{equation}
\end{corollary}
\begin{proof}
To describe our protocol, 
we fix unitaries $V_i:A_iA_i'\rightarrow X_iA_i''$ and then can write the instrument as a CPTP map 
$\cT_i(\alpha) = \sum_{x=1}^{\ell_i} (\proj{x}\ox\1_{A_i''})V_i\alpha V_i^\dagger(\proj{x}\ox\1_{A_i''})$. Its Choi state $\tau^{(i)}_{A_iA_i':X_iA_i''}$ has conditional R\'enyi entropy $\widetilde{H}_2(A_iA_i'|X_iA_i'')_{\tau^{(i)}} = -\log d_i$ \cite{Dupuis-et-al:decouple}. We can thus apply Theorem \ref{theorem:with_smoothing} with Corollary \ref{corollary:D_I-main}, which tell us that there exist local unitaries $U_i$ on $A_i$ such that 
\[\begin{split}
  \sigma_{X_1\ldots X_k A_1''\ldots A_k''R} 
  &= (\cT_1\circ\cU_1\ox\cdots\ox\cT_k\circ\cU_k\ox\id_R) \left(\psi_{A_1\ldots A_k R}\ox\frac{\1_{A_1'}}{c_1}\ox\cdots\ox\frac{\1_{A_k'}}{c_k}\right) \\
  &= \sum_{x_{[k]}} p(x_{[k]}) \proj{x_1}^{X_1}\ox\cdots\ox\proj{x_k}^{X_k}\ox \sigma^{(x_{[k]})}_{A_1''\ldots A_k''R}
\end{split}\]
satisfies
\begin{equation}\begin{split}
  \label{eq:q-SW-decoupling-bound}
  \frac12 &\left\| \sigma_{X_{[k]}A_{[k]}''R} - \frac{\1_{X_1A_1''}}{\ell_1 d_1}\ox\cdots\ox\frac{\1_{X_kA_k''}}{\ell_k d_k}\ox\psi_R \right\|_1 \\
  &\phantom{======}
   \leq 3^k\epsilon + \frac12 \sum_{\emptyset\neq I\subseteq[k]} \exp_2\left[\frac12\!\left(\sum_{i\in I}\log d_i - \sum_{i\in I}\log c_i - H_{\min}^\epsilon(A_I|R)_\psi\right)\right],
\end{split}\end{equation}
choosing all $\epsilon_I=\epsilon$ equal. The right hand side of this bound is $\leq\delta := (3^k+2^{k-1})\epsilon$ if Equation~\eqref{eq:one-shot-merging} is fulfilled. In that case, the total variational distance between $p(x_{[k]})$ and the uniform distribution on $X^k$ is upper bounded by $\delta$, too, and so by the triangle inequality we get
\[
  \sum_{x_{[k]}} p(x_{[k]}) \frac12 \left\| \sigma^{(x_{[k]})}_{A_{[k]}''R} - \frac{\1_{A_1''}}{d_1}\ox\cdots\ox\frac{\1_{A_k''}}{d_k}\ox\psi_R \right\|_1 
  =: \sum_{x_{[k]}} p(x_{[k]}) \delta(x_{[k]})
  \leq 2\delta. 
\]
Notice that $\sigma^{(x_{[k]})}_{A_{[k]}''R} = \Tr_B \proj{\psi^{(x_{[k]})}}_{A_{[k]}''BB_{[k]}'R}$, with 
\[
  \ket{\psi^{(x_{[k]})}}_{A_{[k]}''BB_{[k]}'R} = \frac{1}{\sqrt{p(x_{[k]})}} \bra{x_{[k]}}(V_1U_1\ox\cdots\ox V_kU_k) \left(\ket{\psi}_{A_{[k]}BR}\ox\ket{\Phi}_{A_1'B_1'}\ldots\ket{\Phi}_{A_k'B_k'}\right),
\]
while
\[
  \frac{\1_{A_1''}}{d_1}\ox\cdots\ox\frac{\1_{A_k''}}{d_k}\ox\psi_R = \Tr_{B_1''\ldots B_k'' \widehat{A}_{[k]}\widehat{B}} \Phi_{A_1''B_1''}\ox\cdots\ox\Phi_{A_k''B_k''}\ox\psi_{\widehat{A}_{[k]}\widehat{B}R}.
\]
Then, just as before, we can conclude using Uhlmann's Theorem \ref{theorem:Uhlmann} and the Fuchs-van de Graaf inequalities \eqref{eq:FvdG_relation}, that for each $x_{[k]}$ there exists an isometry $W^{(x_{[k]})}:BB_{[k]}' \rightarrow \widehat{A}_{[k]}\widehat{B}B_{[k]}''$ such that 
\[\begin{split}
  \frac12&\left\| W^{(x_{[k]})}\proj{\psi^{(x_{[k]})}}_{A_{[k]}''BB_{[k]}'R}W^{(x_{[k]})\dagger} \right.\\
  &\phantom{=========}\Biggl.
   -\, \psi_{\widehat{A}_1\ldots \widehat{A}_k\widehat{B}R} \ox (\Phi_{d_1})_{A_1''B_1''} \ox \cdots \ox (\Phi_{d_k})_{A_k''B_k''} \Biggr\|_1 \leq \sqrt{\delta(x_{[k]})(2-\delta(x_{[k]}))}.
\end{split}\]
This means that defining $\cE^{(x_i)}(\alpha) = \bra{x_i} V_iU_i\alpha U_i^\dagger V_i^\dagger \ket{x_i}$ and $\cD^{(x_{k})}(\beta) = W^{(x_{[k]})}\beta W^{(x_{[k]})\dagger}$ as the encoding and decoding maps, this will satisfy the requirement for state merging with error
\[\begin{split}
  \frac12&\left\| \Lambda\left(\psi_{A_1\ldots A_kBR} \ox (\Phi_{c_1})_{A_1'B_1'} \ox \cdots \ox (\Phi_{c_k})_{A_k'B_k'}\right) \right.\\
  &\phantom{======}\left.
   \, - \psi_{\widehat{A}_1\ldots \widehat{A}_k\widehat{B}R} \ox (\Phi_{d_1})_{A_1''B_1''} \ox \cdots \ox (\Phi_{d_k})_{A_k''B_k''} \right\|_1 \leq \sqrt{2\delta(2-2\delta)}
   \leq 4\cdot 3^{k/2} \sqrt{\epsilon}. 
\end{split}\]
With the one-shot achievability in hand, we can now once again use the AEP Theorem \ref{theorem:AEP} for the min-entropy to get the optimal rate region for the i.i.d.~asymptotics of a source $\psi^{\ox n}$ as $n\rightarrow\infty$ and $\delta\rightarrow 0$. Namely, rates $R_i$, defined as the limits of $\frac{r_i}{n}$, are achievable if and only if for all $I\subseteq[k]$, $\sum_{i\in I} R_i \geq S(A_I|A_{I^c}B)_\psi$. This completes the proof of Theorem \ref{theorem:merging} and Corollary \ref{corollary:merging}, since the converse (necessity of the asymptotic inequalities) was argued in \cite{HOW:merging-CMP}.
\end{proof}

To be sure, the achievability of \eqref{eq:q-SW-iid} was shown in \cite{HOW:merging-CMP}, already, by finding the extreme points of the region and noting that they can be solved by iteration of the single-sender merging protocol, and then time-sharing (convex hull) for the remaining region. 
The present protocol (removing the need for time-sharing) was first proposed in the multiple-sender setting by Dutil and Hayden \cite{DutilHayden:merging}, where however the proof of its functioning is incomplete. In Dutil's PhD thesis \cite[Ch.~4]{Dutil:PhD}, the role of simultaneous smoothing is fully analysed. Indeed, a decoupling bound of the form \eqref{eq:q-SW-decoupling-bound} was conjectured there \cite[Conj.~4.1.3]{Dutil:PhD}, and the simultaneous smoothing problem was highlighted. It could be solved only in the i.i.d.~asymptotics of $k=2$ senders. 

To demonstrate a case where the direct attainability of points in the above rate region, and also the one-shot result are relevant, we consider the problem of i.i.d.~state merging for a compound source, i.e.~the source is only partially known, meaning $\rho = \Tr_R \psi \in \cS \subseteq \cS(A_{[k]}B)$, and we would like to design protocols as above for every $n$ that are universal for all $\ket{\psi}^{\ox n}\in A_{[k]}^nB^nR^n$ with $\Tr_R\psi\in\cS$.

\begin{theorem}
  \label{thm:compound-SW}
  In the i.i.d.~limit of $n\rightarrow \infty$, the region of achievable rates $R_i=\frac1n r_i$ for a compound source $\left(\psi^{\ox n}:\Tr_R\psi\in\cS\right)$ is given by
  \begin{equation}
  \label{eq:q-SW-compund}
  \forall I\subseteq[k] \quad \sum_{i\in I} R_i \geq \sup_{\rho\in\cS} S(A_I|A_{I^c}B)_\rho. 
  \end{equation}
\end{theorem}
\begin{proof}
Even when the source $\rho\in\cS$ is fixed, necessarily $\sum_{i\in I} R_i \geq S(A_I|A_{I^c}B)_\rho$ \cite{HOW:merging-CMP}, thus $\sum_{i\in I} R_i \geq \sup_{\rho\in\cS} S(A_I|A_{I^c}B)_\rho$ for all subsets $I$. This takes care of the converse bound in Equation~\eqref{eq:q-SW-compund}, and it remains to prove the achievability. 

To this end, for block length $n$ we choose an $\frac{\eta}{n}$-net $\cS_0 \subset \cS$ of states for $\cS$ (i.e.~a net to approximate elements of $\cS$). By adapting the proof of Lemma \ref{lemma:net}, we find that  $N := |\cS_0| \leq \left(\frac{5n}{\eta}\right)^{2|A_{[k]}|^2|B|^2}$. We number the elements of the net, $\cS_0 = \{\rho_s : s=1,\ldots,N\}$ and choose purifications $\ket{\psi}_s \in A_{[k]}BR$ of $\rho_s$. As previously deployed, the plan is to construct a one-shot protocol for the averaged source 
\begin{align*}
  \widetilde{\rho} &= \frac{1}{N} \sum_{s=1}^N \rho_s^{\ox n} \text{ on } A_{[k]}^nB^n, \text{ which has a purification } \\
  \ket{\widetilde{\psi}} &= \frac{1}{\sqrt{N}} \sum_{s=1}^N \ket{\psi_s}^{\ox n}\ox\ket{s}_{R'} \in A_{[k]}^nB^nR^nR',
\end{align*}
then argue that the protocol performs well on all $\psi_s^{\ox n}$, and finally that it must perform well on all $\psi^{\ox n}$ with $\Tr_R \psi \in \cS$. We could do this directly using Theorem \ref{theorem:merging}, except that for that to work we have to make the smoothing parameter $\epsilon$ in the min-entropies dependent on $n$, which makes the argument awkward. Instead, we opt to use the R\'enyi decoupling from Theorem \ref{theorem:without_smoothing} (Corollary \ref{corollary:D_I-main}), following otherwise the proof of Theorem \ref{theorem:merging}. This means that there,  Equation~\eqref{eq:q-SW-decoupling-bound} is replaced by 
\begin{equation}\begin{split}
  \label{eq:q-SW-decoupling-bound-Renyi}
  \frac12 &\left\| \sigma_{X_{[k]}A_{[k]}''R^nR'} - \frac{\1_{X_1A_1''}}{\ell_1 d_1}\ox\cdots\ox\frac{\1_{X_kA_k''}}{\ell_k d_k}\ox\widetilde{\psi}_{R^nR'} \right\|_1 \\
  &\phantom{======}
   \leq \sum_{\emptyset\neq I\subseteq[k]} 
          \exp_2\left[ \frac{\alpha-1}{\alpha} \left(\sum_{i\in I}\log d_i - \sum_{i\in I}\log c_i - \widetilde{H}_{\alpha}(A_I^n|R^nR')_{\widetilde{\psi}} \right)\right],
\end{split}\end{equation}
choosing all $\alpha_I=\alpha \in (1,2]$ equal. With the net rates $nR_i = \log c_i - \log d_i$, the right hand side of the last bound is $\leq \delta$ if 
\[\begin{split}
  \forall \emptyset\neq I\subseteq[k] \quad
  \sum_{i\in I} nR_i &\geq -\widetilde{H}_{\alpha}(A_I^n|R^nR')_{\widetilde{\psi}} - \frac{\alpha}{\alpha-1}\log\left(2^{-k}\delta\right) \\
  &= \widetilde{H}_{\beta}(A_I^n|A_{I^c}^nB^n)_{\widetilde{\psi}} - \frac{\beta}{1-\beta}\log\left(2^{-k}\delta\right),
\end{split}\]
where we have used the R\'enyi entropy duality (Lemma \ref{lemma:Renyi-duality}) with $\frac{1}{\beta}+\frac{1}{\alpha}=2$. In fact, we can simplify this condition using Lemma \ref{lemma:quasi-concavity} which tells us
\[\begin{split}
  \widetilde{H}_{\beta}(A_I^n|A_{I^c}^nB^n)_{\widetilde{\psi}} 
   &\leq \max_{s\in[N]} \widetilde{H}_{\beta}(A_I^n|A_{I^c}^nB^n)_{\rho_s^{\ox n}} + \log N \\
  &\leq \sup_{\rho\in\cS} \widetilde{H}_{\beta}(A_I^n|A_{I^c}^nB^n)_{\rho^{\ox n}} + \log N \\
  &= \sup_{\rho\in\cS} n\widetilde{H}_{\beta}(A_I|A_{I^c}B)_{\rho} + \log N. 
\end{split}\]
Thus, the trace norm in Equation~\eqref{eq:q-SW-decoupling-bound-Renyi} is $\leq\delta$ if
\[
  \forall\emptyset\neq I \subseteq[k] \quad
    \sum_{i\in I} R_i \geq \sup_{\rho\in\cS} \widetilde{H}_{\beta}(A_I|A_{I^c}B)_{\rho} + \frac1n \left(\log N + \frac{\beta}{1-\beta}\left(k-\log\delta \right)\right). 
\]
Since $\widetilde{H}_\beta(A_I|A_{I^c}B)_{\rho}$ converges to $S(A_I|A_{I^c}B)_{\rho}$ as $\beta\rightarrow 1$, and the converging as well as the limit functions are continuous on the compact set of all states, hence uniformly continuous, also the convergence  $\widetilde{H}_\beta(A_I|A_{I^c}B) \rightarrow S(A_I|A_{I^c}B)$ of the functions on state space is uniform. Thus, there exists a $\Delta(\beta)>0$ (converging to $0$ as $\beta\rightarrow 1$) such that for all $I\subseteq[k]$, 
\[
  \sup_{\rho\in\cS} S(A_I|A_{I^c}B)_{\rho} 
    \leq \sup_{\rho\in\cS} \widetilde{H}_\beta(A_I|A_{I^c}B)_{\rho} 
    \leq \sup_{\rho\in\cS} S(A_I|A_{I^c}B)_{\rho} + \Delta(\beta).
\]
And so the trace norm in Equation~\eqref{eq:q-SW-decoupling-bound-Renyi} is guaranteed to be $\leq\delta$ if
\begin{equation}
  \label{eq:compound-rate}
  \forall\emptyset\neq I \subseteq[k] \quad
    \sum_{i\in I} R_i \geq \sup_{\rho\in\cS} S(A_I|A_{I^c}B)_{\rho} + \Delta(\beta) + \frac1n \left(\log N + \frac{\beta}{1-\beta}\left(k-\log\delta \right)\right). 
\end{equation}

Continuing the reasoning of the proof of Theorem \ref{theorem:merging}, we obtain a merging protocol for $\widetilde{\psi}$ that has error $\leq 2\sqrt{\delta}$, hence it has error $\leq 2N\sqrt{\delta}$ on each of the $\psi_s^{\ox n}$, and so finally error $\leq 2N\sqrt{\delta}+2\eta$ on each of source $\psi^{\ox n}$ such that $\Tr_R\psi\in\cS$. Choosing $\delta=\frac{\eta^2}{N^2}$ we get an error guarantee of $\leq 4\eta$ across the set $\cS$, while the rates are bounded 
\begin{equation*}
  \forall\emptyset\neq I \subseteq[k] \quad
    \sum_{i\in I} R_i \geq \sup_{\rho\in\cS} S(A_I|A_{I^c}B)_{\rho} + \Delta(\beta) + O\left( \frac{\log n - \log\eta}{n(1-\beta)}\right), 
\end{equation*}
which for $n\rightarrow\infty$ and $\beta\rightarrow 1$ proves the claim. 
\end{proof}

\subsection{Quantum communication via quantum multiple access channels}
\label{subsec:MAC}
A quantum multiple access channel is a CPTP map $\cN:A_1\ldots A_k \rightarrow B$ from $k$ senders $A_i$ to a single receiver $B$. For later use, let us introduce the Stinespring dilation $\cN(\rho) = \Tr_E V \rho V^\dagger$, with $V:A_1\ldots A_k \rightarrow BE$ an isometry. Let each user $i$ hold independent quantum messages (quantum systems) $M_i$ of dimension $s_i=|M_i|$. Then, a code for such a channel consists of a set of encoding CPTP maps $\cE_i:M'_i\rightarrow A'_i$ and a single decoding CPTP map $\cD:B\rightarrow \widehat{M}_1\ldots\widehat{M}_k$ where $\widehat{M}_i \simeq M_i$. And the numbers $\log s_i$ are the one-shot rates. In this setting, we say that the code has error $\delta$ if 
\[
  \frac12 \left\| \left(\cD\!\circ\!\cN\!\circ\!(\cE_1\!\ox\cdots\ox\!\cE_k)\!\ox\id_{M_{[k]}}\!\right)\! (\Phi_{M_1'M_1}\!\ox\cdots\ox\!\Phi_{M_k'M_k}) 
  - \Phi_{\widehat{M}_1M_1}\!\ox\cdots\ox\Phi_{\widehat{M}_k M_k} \right\|_1 
  \leq \delta,
\]

where $\Phi_{M_i'M_i}$ and $\Phi_{\widehat{M}_i M_i}$ are standard maximally entangled states of Schmidt rank $s_i$. The problem here is now to characterize, for a given error $\delta$, the set of achievable one-shot rate tuples $(\log s_1, \ldots, \log s_k)$. Likewise, in the i.i.d.~asymptotic limit $\cN^{\ox n}$, when $n\rightarrow\infty$ and $\delta\rightarrow 0$, we introduce the asymptotic rates $R_i = \frac1n \log s_i$ and ask for a description of the achievable rate tuples $(R_1,\ldots,R_k)$. By general principles this is a convex corner, i.e.~a closed convex set in the positive orthant, containing the origin and stable under reducing any coordinate towards $0$. See Figure \ref{fig:3DrrMAC} for the one-shot tripartite rate region and Figure \ref{fig:iidRRMAC} for the i.i.d.~bipartite rate region.
\begin{theorem}
  \label{theorem:MAC}
  Given the quantum MAC $\cN:A_{[k]} \rightarrow B$ and its Stinespring isometry $V:A_{[k]} \rightarrow BE$, as well as pure states $\varphi^{(i)}_{A_iA_i'}$ with $A_i' \simeq A_i$ ($i\in[k]$), define
  \begin{equation}
    \label{eq:MAC-reference-state}
    \ket{\psi}_{A_1\ldots A_kBE} = (\1_{A_{[k]}}\ox V)\left( \ket{\varphi^{(1)}}_{A_1A_1'}\ox\cdots\ox\ket{\varphi^{(k)}}_{A_kA_k'} \right),
  \end{equation}
  where we let $V$ act on $A_{[k]}'$.
  Then there exists a code for the channel with error $\delta \leq (k+1)2^{k+1}\sqrt{\epsilon}$ if the one-shot rate tuples satisfy
  \begin{equation}
    \label{eq:one-shot-MAC}
    \forall \emptyset\neq I\subseteq[k] 
    \quad 
    \sum_{i\in I} \log s_i \leq H_{\min}^\epsilon(A_I|E)_\psi + 2\log\epsilon
    = -H_{\max}^\epsilon(A_I|A_{I^c}B)_\psi + 2\log\epsilon. 
\end{equation}
\end{theorem}

\begin{figure}[ht]
    \centering
    \includegraphics[scale=0.7]{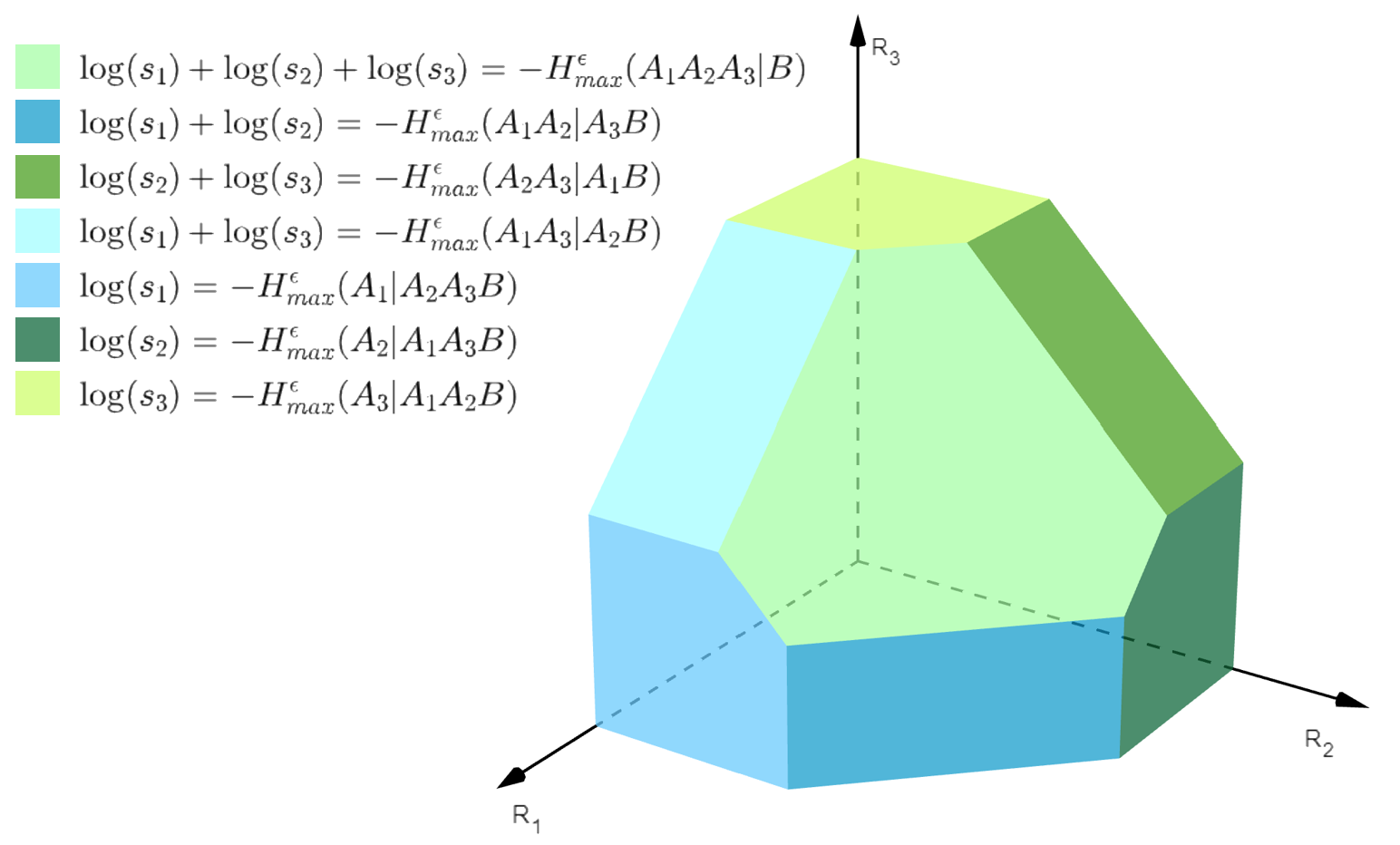}
    \caption{One-shot achievable rate region for a MAC with three senders $A_1$, $A_2$ and $A_3$.}
    \label{fig:3DrrMAC}
\end{figure}

\begin{corollary}
  \label{corollary:MAC}
  In the i.i.d.~asymptotic limit $n\rightarrow\infty$ and $\delta\rightarrow 0$, the rates $R_i=\frac{1}{n}\log s_i$ are achievable for transmission over $\cN^{\ox n}$ if
  \begin{equation}
    \label{eq:iid-MAC}
    \forall I\subseteq[k] \quad
    \sum_{i\in I} R_i \leq I(A_I\rangle BA_{I^c})_\psi,
  \end{equation}
  where $I(A_I\rangle BA_{I^c})_\psi = - S(A_I|BA_{I^c})_\psi $ is the coherent information. More generally, for an ensemble $\{q(u),\ket{\psi_u}\}$ of states as in Equation~\eqref{eq:MAC-reference-state}, $u\in\cU$ ranging over a discrete alphabet, the rates $R_i$ are achievable if 
  \begin{equation}
    \label{eq:iid-MAC-convexhull}
    \forall I\subseteq[k] \quad
    \sum_{i\in I} R_i \leq \sum_u q(u) I(A_I\rangle BA_{I^c})_{\psi_u} = I(A_I\rangle BA_{I^c}U)_{\overline{\psi}}, 
  \end{equation}
  the latter coherent information evaluated on the cq-state $\overline{\psi} = \sum_u q(u)\proj{u}_U \ox \psi_u$.
\end{corollary}
\begin{proof}
We prove both Theorem \ref{theorem:MAC} and Corollary \ref{corollary:MAC}.
To describe good codes, fix projective measurements $(P_{j^{(i)}})$ on $A_i$, where each of the $P_{j^{(i)}}$ has rank $s_i$
(by enlarging $A_i$ if necessary, we may assume w.l.o.g. that $s_i$ divides the dimension $|A_i|$), and let the corresponding CPTP map be
$\cT_i(\alpha) = \sum_{j^{(i)}=1}^{|A_i|/s_i} P_{j^{(i)}} \alpha P_{j^{(i)}}$. Its Choi state $\tau^{(i)}_{A_i'A_i}$ has conditional R\'enyi entropy $\widetilde{H}_2(A_i'|A_i)_{\tau^{(i)}} = -\log s_i$ \cite{Dupuis-et-al:decouple}. We can thus apply Theorem \ref{theorem:with_smoothing} with Corollary \ref{corollary:D_I-main} that tell us that there exist local unitaries $U_i$ on $A_i$ such that 
\[
  \sigma_{A_1\ldots A_k E} = (\cT_1\circ\cU_1\ox\cdots\ox\cT_k\circ\cU_k\ox\id_E)\psi_{A_1\ldots A_k E}
\]
satisfies
\[
  \frac12 \left\| \sigma_{A_{[k]} E} - \frac{\1_{A_1}}{|A_1|}\!\ox\cdots\ox\!\frac{\1_{A_k}}{|A_k|}\ox\psi_E \right\|_1 
  \!\!\leq 3^k\epsilon + \frac12 \sum_{\emptyset\neq I\subseteq[k]}\!\!\! \exp_2\left[\frac12\!\left(\sum_{i\in I}\log s_i - H_{\min}^\epsilon(A_I|E)_\psi\!\right)\right],
\]
choosing all $\epsilon_I=\epsilon$ equal, the right-hand side is $\leq\eta := (3^k+2^{k-1})\epsilon$ if Equation~\eqref{eq:one-shot-MAC} is fulfilled. In that case, there must exist measurement outcomes $j^{(i)}$ for the POVM on $A_i$ such that, with the outcome probability $p(\vec{j}) = \Tr (U_1\ox\cdots\ox U_k)\psi_{A_1\ldots A_k}(U_1\ox\cdots\ox U_k)^\dagger \left(P_{j^{(1)}}\ox\cdots\ox P_{j^{(k)}}\right)$: 
\[\begin{split}
  \frac12 &\left\| \frac{1}{p(\vec{j})} \left(P_{j^{(1)}}U_1\ox\cdots\ox P_{j^{(k)}}U_k\ox\1_E\right) \psi_{A_1\ldots A_k E} \left(U_1^\dagger P_{j^{(1)}}\ox\cdots\ox U_k^\dagger P_{j^{(k)}}\ox\1_E\right) \right.\\
  &\phantom{=============================}\left. 
   - \frac{P_{j^{(1)}}}{s_1}\ox\cdots\ox\frac{P_{j^{(k)}}}{s_k}\ox\psi_E \right\|_1 \leq 2\eta. 
\end{split}\]
Unpacking the definition of $\psi$, this means that for each $i$ there are unit vectors of Schmidt rank $s_i$, $\ket{\widetilde{\varphi}^{(i)}}_{A_i'A_i} \propto (P_{j^{(i)}}U_i\ox\1)\ket{\varphi^{(i)}}$, such that for all $\vec{j}$
\[\begin{split}
  \frac{1}{\sqrt{p(\vec{j})}} &\left(P_{j^{(1)}}U_1\ox\cdots\ox P_{j^{(k)}}U_k\ox\1_B\ox\1_E\right) \ket{\psi}_{A_1\ldots A_k BE} \\
  &\phantom{==========}
   = (\1_{A_1}\ox\cdots\ox\1_{A_k}\ox V) \left(\ket{\widetilde{\varphi}^{(1)}}_{A_1'A_1}\cdots\ket{\widetilde{\varphi}^{(k)}}_{A_k'A_k}\right).
\end{split}\]
The previous trace norm estimate shows that each of the $\widetilde{\varphi}^{(i)}_{A_i}$ is nearly maximally mixed on its support $M_i$, up to trace distance $\leq 2\eta$. So using Uhlmann´s theorem \ref{theorem:Uhlmann} and the inequalities \eqref{eq:FvdG_relation} once again we conclude that there must exist maximally entangled states $\Phi_{M_i'M_i}$ of Schmidt rank $s_i$ and isometries $W_i : M_i' \rightarrow A_i'$, such that $P\left( (W_i\ox\1_{M_i})\Phi_{M_i'M_i}(W_i\ox\1_{M_i})^\dagger, \widetilde{\varphi}^{(i)} \right) \leq 2\sqrt{\eta}$. Putting these last bounds together, using the triangle inequality for the purified distance and its non-increase under CPTP maps, we get

\[\begin{split}
  P&\left( \Tr_B V^{A_{[k]}'\rightarrow BE} \left( \bigotimes_{i=1}^k W_i^{M_i'\rightarrow A_i'}\Phi_{M_i'M_i} (W_i^{M_i'\rightarrow A_i'})^\dagger \right) (V^{A_{[k]}'\rightarrow BE})^\dagger , \right. \\
  &\phantom{==========================}
   \left. \frac{\1_{M_1}}{s_1}\ox\cdots\ox\frac{\1_{M_k}}{s_k}\ox\psi_E \right) \leq 2(k+1)\sqrt{\eta}.
\end{split}\]
As the second argument in the purified distance has purification $\Phi_{\widehat{M_1}}\ox\cdots\ox\Phi_{\widehat{M_k}M_k}\ox \psi_{A_{[k]}BE}$, by Uhlmann's theorem there exists an isometry $\widehat{W}:B\rightarrow \widehat{M}_1\ox\cdots\ox\widehat{M}_k\ox A_{[k]}B$ such that 
\[\begin{split}
  P&\left( \widehat{W} V^{A_{[k]}'\rightarrow BE} \left( \bigotimes_{i=1}^k W_i^{M_i'\rightarrow A_i'}\Phi_{M_i'M_i} (W_i^{M_i'\rightarrow A_i'})^\dagger \right) (V^{A_{[k]}'\rightarrow BE})^\dagger \widehat{W}^\dagger, \right. \\
  &\phantom{=====================}
   \Bigg. \Phi_{\widehat{M_1}M_1}\ox\cdots\ox\Phi_{\widehat{M_k}M_k}\ox \psi_{A_{[k]}BE} \Biggr) \leq 2(k+1)\sqrt{\eta}.
\end{split}\]
In other words, defining the encoders $\cE_i(\alpha) = W_i\alpha W_i^\dagger$ and the decoder as $\cD(\beta) = \Tr_{A_{[k]}B} \widehat{W}\beta\widehat{W}^\dagger$, yields a code for the quantum MAC with one-shot rates $s_i$ [subject to the conditions \eqref{eq:one-shot-MAC}] and error $\delta = 2(k+1)\sqrt{\eta} \leq (k+1)2^{k+1}\sqrt{\epsilon}$. This form of a one-shot achievability region had been conjectured for a long time, with the best previous result reported by Chakraborty, Nema, and Sen \cite{chakraborty-et-al:multisender.decoupling,CNS:one-shotMAC}, who used rate-splitting and a multipartite decoupling with a modified smooth collision entropy. Using the encoder and decoder defined above, we can attain any point in the one-shot capacity region in \eqref{eq:one-shot-MAC}.

As in the previous example applications, we can directly apply the AEP Theorem \ref{theorem:AEP} for the min-entropy \cite{Tomamichel:PhD} to obtain an achievable rate region for the i.i.d.~quantum multiple-access channel: $H_{\min}^\epsilon(A_I^n|E^n)_{\psi^{\ox n}} \sim n S(A_I|E)_\psi = -n S(A_I|BA_{I^c})_\psi = n I(A_I\rangle BA_{I^c})_\psi$, the latter quantity being the coherent information. Then, rates $R_i = \frac{1}{n}\log s_i$ are achievable in the limit $n\rightarrow\infty$ and $\delta\rightarrow 0$ if Equation~\eqref{eq:iid-MAC} is satisfied. The more general statement with the distribution $q$ over $u$ is obtained by applying the AEP to the tensor product $\bigotimes_{u\in\cU} \psi_u^{n_u}$, where $n_u$ are non-negative integers with $\sum_u n_u = n$ and $\sum_u \left| \frac{n_u}{n}-q(u) \right| \rightarrow 0$. This completes the proof.
\end{proof}

\begin{figure}[ht]
    \centering
    \includegraphics[scale=0.7]{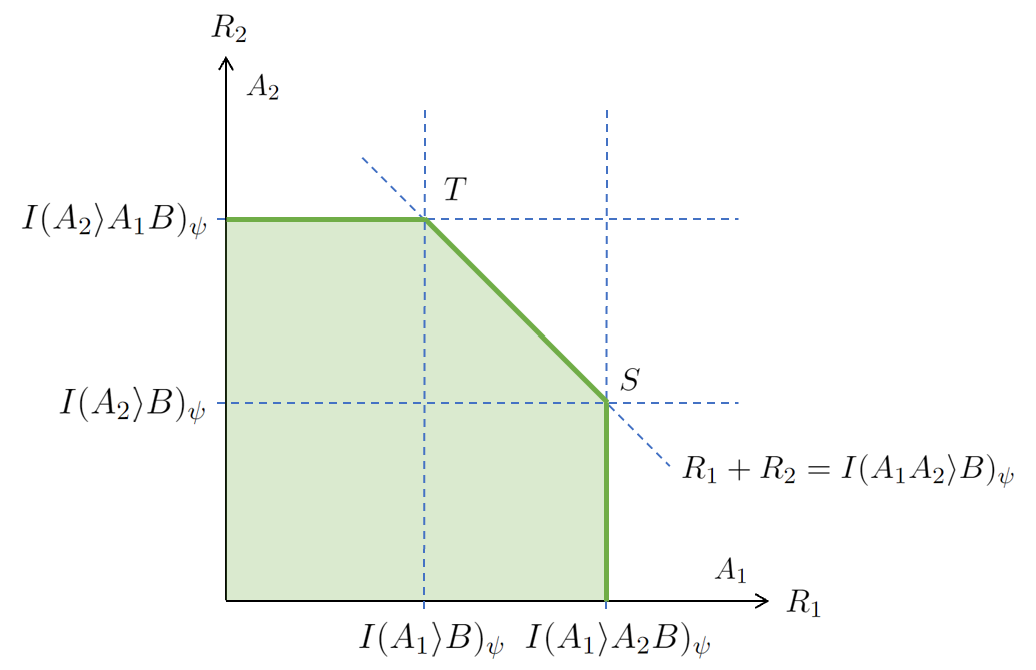}
    \caption{Achievable rate region of a MAC with two senders $A_1$ and $A_2$ in the i.i.d.~limit.}
    \label{fig:iidRRMAC}
\end{figure}

This rate region inner bound goes back to Yard, Devetak and Hayden \cite{YDH:MAC}, where it was obtained by determining the extremal points of the above region, attaining these by successive decoders and the rest of the region by time-sharing (convex combination of rates). In the two-sender case (see Fig.~\ref{fig:iidRRMAC}) these extremal points are $T=[I(A_1\rangle B)_\psi,I(A_2\rangle A_1B)_\psi]$ and $S=[I(A_1\rangle A_2B)_\psi,I(A_2\rangle B)_\psi]$. In the present proof we can achieve for the first time each point of the region directly by a quantum simultaneous decoder, and without needing to appeal to the simultaneous smoothing conjecture (cf. \cite{chakraborty-et-al:multisender.decoupling}). 

As an illustration of a situation where it is essential to reach each point in the convex hull of the corner directly and without time-sharing, we solve the problem of communication via a \emph{compound channel}, which is given by a subset $\mathfrak{C} \subset \text{CPTP}(A_{[k]}\rightarrow B)$ of the quantum channels mapping the $A_i$ to $B$. A code of block length $n$ for the compound channel is defined as above, but the error is the supremum over the error when applying the code to $\cN^{\ox n}$, $\cN\in\mathfrak{C}$.

Inspired by Mosonyi's approach to the single-sender case of classical communication \cite{Mosonyi}, using the R\'enyi decoupling bound (Theorem \ref{theorem:without_smoothing} and Corollary \ref{corollary:D_I-main}),
we can prove the following general achievability result. 
\begin{theorem}
  \label{thm:compound-MAC}
  Given the compound channel $\mathfrak{C} \subset \text{CPTP}(A_{[k]}\rightarrow B)$, a probability distribution $q(u)$ over a discrete alphabet and reference states $\ket{\varphi_u^{(i)}} \in A_iA_i'$ ($i\in[k]$), define the states
  \[
    \rho_u(\cN) = (\id_{A_{[k]}}\ox \cN)\left( \ket{\varphi_u^{(1)}}_{A_1A_1'}\ox\cdots\ox\ket{\varphi_u^{(k)}}_{A_kA_k'} \right),
  \]
  where we let $\cN \in\mathfrak{C}$ act on $A_{[k]}'$.
  Then the asymptotic rates $R_i$ are achievable if 
  \begin{equation*}
    \forall I\subseteq[k] \quad
    \sum_{i\in I} R_i 
      \leq \inf_{\cN\in\mathfrak{C}} \sum_u q(u) I(A_I\rangle BA_{I^c})_{\rho_u(\cN)}
      = \inf_{\cN\in\mathfrak{C}} I(A_I\rangle BA_{I^c}U)_{\overline{\rho}(\cN)}, 
  \end{equation*}
  the latter coherent information evaluated on the cq-state $\overline{\rho}(\cN) = \sum_u q(u)\proj{u}_U \ox \rho_u(\cN)$.
\end{theorem}

The proof combines the ideas of Theorem \ref{theorem:MAC} and Corollary \ref{corollary:MAC} applied to the uniform mixture channel $\widetilde{\cN} = \frac{1}{N} \sum_{t=1}^N \cN_t^{\ox n}$ over a net for the set $\mathfrak{C}$ (with respect to the diamond norm), and proceeds like the analogous proof of Theorem \ref{thm:compound-SW} in the previous subsection on compound quantum state merging, and we thus omit the details.


\section{Discussion}
\label{sec:conclusion}
Decoupling is a fundamental primitive in the design of quantum transmission codes, quantum Slepian-Wolf coding, cryptographic communication, and channel simulation, but has so far been largely limited to single-user settings. Here we have shown how to leverage tensorisation properties of expected-contractive maps, to extend the basic toolbox to simultaneous decoupling in a multipartite setting where each party applies their own random unitary. We have managed to find achievability bounds for general multipartite decoupling in terms of smooth conditional min-entropies as usual in one-shot scenarios (Theorem \ref{theorem:with_smoothing}); and in terms of conditional R\'enyi entropies (Theorem \ref{theorem:without_smoothing}). 

Our approach should be contrasted with the ``standard'' one of passing to a Hilbert-Schmidt norm bound already in the first line of Equation \eqref{eq:split-into-Theta}, seeing that we can evaluate quadratic averages not only of single random unitaries but also their tensor products. This has been done in \cite{DutilHayden:merging,Dutil:PhD} and \cite{chakraborty-et-al:multisender.decoupling}, and perhaps by other authors who have found themselves then at the same impasse. For simplicity, consider a tripartite quantum state $\rho_{A_1A_2E}$ (i.e.~$k=2$) and the usual setup of the composition of local unitary operations ($U_1$ on $A_1$ and $U_2$ on $A_2$) followed by a fixed CPTP map $\mathcal{T}_{A_1A_2\rightarrow B}$ with Choi matrix $\tau_{A_1A_2B}$. We can use Lemma \ref{lemma:tr_norm_bound} to bound
\[\begin{split}
    &\left\| \mathcal{T}_{A_1A_2\rightarrow B}[(U_1\otimes U_2)\rho_{A_1A_2E}(U_1\otimes U_2)^\dagger]-\tau_B\otimes\rho_E \right\|_1^2 \\
    &\phantom{=====}
     \leq \Tr\left[\left((\sigma\otimes\zeta)^{-1/4}(\mathcal{T}[(U_1\otimes U_2)\rho(U_1\otimes U_2)^\dagger]-\tau_B\otimes\rho_E)(\sigma\otimes\zeta)^{-1/4}\right)^2\right], 
\end{split}\]
for two auxiliary states $\sigma_B$ and $\zeta_E$. At this point we have passed already to the trace of a square, and following the method in \cite{Dupuis-et-al:decouple} and used above (see also \cite{chakraborty-et-al:multisender.decoupling}), 
we define $\widetilde{\mathcal{T}}_{A_1A_2\rightarrow B}(\cdot)=\sigma^{-1/4}\mathcal{T}_{A_1A_2\rightarrow B}(\cdot)\sigma^{-1/4}$, $\tilde{\rho}_{A_1A_2E}=\zeta_E^{-1/4}\rho_{A_1A_2E}\zeta_E^{-1/4}$ and also $\tilde{\tau}_{A_1A_2B}=\sigma_B^{-1/4}\tau_{A_1A_2B}\sigma_B^{-1/4}$. Then, after expanding the square, evaluating the expectations using $\int {\rm d}U\, UXU^\dagger = (\Tr X)\frac{\1}{d}$ and Corollary \ref{corollary:schur_general}, and after optimizing $\sigma_B$ and $\zeta_E$ we finally get
\[\begin{split}
    \EE_{U_1,U_2}& \left\| \mathcal{T}_{A_1A_2\rightarrow B}[(U_1\otimes U_2)\rho_{A_1A_2E}(U_1\otimes U_2)^\dagger]-\tau_B\otimes\rho_E \right\|_1^2 \\
    &\hspace{-3pt}\leq D\Bigg(\Tr[\tilde{\tau}_{A_2B}^2]\Tr[\tilde{\rho}_{A_2E}^2]+\Tr[\tilde{\tau}_{A_1B}^2]\Tr[\tilde{\rho}_{A_1E}^2]+\Tr[\tilde{\tau}_{A_1A_2B}^2]\Tr[\tilde{\rho}_{A_1A_2E}^2]\Bigg) \\
    &\hspace{-3pt}= D\!\left(2^{-\widetilde{H}_2(A_1|B)_\tau-\widetilde{H}_2(A_1|E)_\rho}+2^{-\widetilde{H}_2(A_2|B)_\tau-\widetilde{H}_2(A_2|E)_\rho}+2^{-\widetilde{H}_2(A_1A_2|B)_\tau-\widetilde{H}_2(A_1A_2|E)_\rho}\right) \\
    &\hspace{-3pt}\leq D\!\left(2^{-\widetilde{H}_2(A_1|B)_\tau-{H}_{\min}(A_1|E)_\rho}\!+\!2^{-\widetilde{H}_2(A_2|B)_\tau-{H}_{\min}(A_2|E)_\rho}\!+\!2^{-\widetilde{H}_2(A_1A_2|B)_\tau-{H}_{\min}(A_1A_2|E)_\rho}\right),
\end{split}\]
where in the last line we have lower bounded the collision entropies by min-entropies, and $D\!=\!2\left(1-\frac{1}{\abs{A_1^2}}\right)^{-1}\!\!\left(1-\frac{1}{\abs{A_2^2}}\right)^{-1}$ is a constant like the ones encountered in Theorems \ref{theorem:with_smoothing} and \ref{theorem:without_smoothing}.

The resulting bound thus has the characteristic sum of exponential terms, one for each subset of parties, and the exponents feature conditional min- and collision entropies of the state and of the fixed channel Choi matrix, respectively, recalling the structure of \cite{Dupuis-et-al:decouple}. So in some sense, this is a one-shot decoupling theorem. The technical problem is that we have left the realm of trace distances in the very first step, and so the min-entropies in the final expression all refer to the same state. 

If now we want to move to smooth min-entropies to optimize the attainable rates we need to smooth the global state so as to approximate all reduced states' smooth min-entropies simultaneously. The long-standing simultaneous smoothing conjecture \cite{DrescherFawzi:sim-min} states that this is possible in some way, but remains unsolved. In \cite{chakraborty-et-al:multisender.decoupling} it is partially addressed to lead to an improved one-shot decoupling bound, but in the application to an i.i.d.~coding problem one still has to appeal to the asymptotic version of the simultaneous smoothing conjecture, which remains open, too. Instead, the innocent-looking step of passing to the second line in Equation~\eqref{eq:split-into-Theta} gains us a sum of tensor product random maps, which we can split up using the triangle inequality so that each term can be dealt with via its own quadratic average bound; at the end, we can then apply smoothing separately to each of the exponential terms corresponding to the subsets of parties. We thus prove the conjectured form of simultaneous local decoupling, while not having to address the simultaneous smoothing conjecture. 

We have shown the power of these results by presenting a series of relevant applications in multi-user quantum information tasks. We have found one-shot, finite block length, and asymptotic achievability results in local randomness extraction, multipartite entanglement distillation, and quantum communication via quantum multiple access channels. 

\begin{itemize} \setlength\itemsep{0mm}
\item In particular, we have found a one-shot version of local randomness extraction and achievability rates for an arbitrary number $k$ of cooperating users, as well as the optimal rate region in the i.i.d asymptotics. The latter result reproduces the core insight of \cite{YHW:randomness} for $k=2$ collaborating parties, albeit with a much simpler protocol, and proves the conjectured rate region for an arbitrary number $k$ of users. 

\item Concerning multi-party entanglement of assistance, we have also found a one-shot and i.i.d.~optimal rates, reproducing the asymptotic results from \cite{HOW:merging-CMP} with a much simpler approach. Actually, the used procedure was previously analyzed in \cite{Dutil:PhD} and shown to work assuming the simultaneous smoothing conjecture. With the application of our theorems, we do not require the use of this unproven conjecture.

\item Likewise, we solve the quantum version of the Slepian-Wolf data compression of correlated sources, which reduces to the task of quantum state merging, in the one-shot setting, as suggested by \cite{DutilHayden:merging,Dutil:PhD}, as well as the i.i.d.~setting, reproducing the asymptotically optimal rate region of \cite{HOW:merging-Nature,HOW:merging-CMP} and proving the conjectured one-shot achievable region, by achieving each point of the respective regions directly, without the need of time-sharing and without the simultaneous smoothing conjecture.

\item Finally, we have found a one-shot achievability region for quantum communication via quantum multiple access channels that had been conjectured for a long time. In a similar fashion to the previous applications, we obtained an achievable rate region for the i.i.d.~quantum MAC, reproducing the result of \cite{YDH:MAC}. For the first time, we can achieve each point of that region directly by a quantum simultaneous decoder, without the need of time-sharing, and without the simultaneous smoothing conjecture.
\end{itemize}

To illustrate the utility of the one-shot results we showed that they also solve the compound source/channel versions of all four problems. These are conceptually important results since they prove that attainable rates are in some sense robust and do not require perfect knowledge of the source/channel. Indeed, consider the important case that the set $\cS$ ($\mathfrak{C}$) is a small trace-norm (diamond-norm) ball around an ``ideal'' state (channel). Then Theorems \ref{thm:compound-randomness}, \ref{thm:compound-EoA}, \ref{thm:compound-SW} and \ref{thm:compound-MAC} in particular imply that the optimal rates of the ideal state/channel can be almost achieved by a protocol that works uniformly well in the whole neighbourhood of the ideal. 
%

An important future problem will be to extend the multipartite randomness extraction model to the cryptographic setting, where typically only lower bounds on the min-entropies $H_{\min}^{\epsilon}(A_I|E)_\rho$ are available. In that case, an extractor needs a seed of randomness to start with. For example, Theorem \ref{theorem:local_rand_extr} (and Theorem \ref{theorem:with_smoothing} on which it is based), requires only a unitary $2$-design to give security guarantees with high probability. That is to say, each local user could use a random element of the Clifford group as a seed. However, schemes with much smaller seeds are known in single-user settings \cite{BertaFawziWehner:randomness,Nakata-et-al:decouple,Vadhan:pseudo}, and it will be interesting to adapt these to the multi-user case.

\section*{Acknowledgements}
The authors thank Fr\'ed\'eric Dupuis for his encouragement to try out the application of the R\'enyi decoupling approach to multi-user problems. We furthermore thank Hao-Chung Cheng, Li Gao, and Mario Berta for exchanging notes about our mutually independent work on decoupling and for sharing their manuscript \cite{ChengGaoBerta:broadcast} prior to making it public.

PC and AW are supported by the Institute for Advanced Study of the Technical University of Munich, by way of a Hans Fischer Senior Fellowship. AW is furthermore supported by the European Commission QuantERA grant ExTRaQT (Spanish MICINN project PCI2022-132965), by the Spanish MINECO (project PID2019-107609GB-I00) with the support of FEDER funds, the Generalitat de Catalunya (project 2017-SGR-1127), the Spanish MICINN with funding from European Union NextGenerationEU (PRTR-C17.I1) and the Generalitat de Catalunya, and by the Alexander von Humboldt Foundation.

\section*{Statements and declarations}
\begin{itemize}
    \item \textbf{Competing interests statement}: The authors have no relevant competing interests to disclose.
    \item \textbf{Data availability statement}: The authors affirm that the primary source of data supporting the findings of this study is contained within the paper, with the cited references serving as supplementary sources.

\end{itemize}

\printbibliography

\end{document}